\documentclass[10pt]{article}
\usepackage[utf8]{inputenc}
\usepackage{color,soul}
\usepackage{fullpage}
\usepackage{csquotes}

\usepackage{dsfont,amsthm, amssymb,amsmath,stmaryrd,authblk,mathtools,tikz,graphicx,pgfplots}
\pgfplotsset{compat=1.14}
\usepgfplotslibrary{external}
\usetikzlibrary{positioning}
\usetikzlibrary{decorations.pathmorphing}
\usepackage{caption,subcaption}

\tikzset{dot/.style 2 args={fill, circle, inner sep=1pt, label={#1:\scriptsize #2}}}

\usepackage{graphicx,color,tikz}

\usepackage[
backend=biber,
bibencoding=utf8,
sorting=ydnt,
sortcites=true
]{biblatex}
\usepackage{relsize}
\usepackage{url}
\usepackage[colorlinks=true, urlcolor=blue,citecolor=blue,anchorcolor=blue]{hyperref}
\usepackage[capitalise]{cleveref}
\usepackage{mymacros}
\usepackage{caption}

\DeclareCaptionFormat{algor}{
  \hrulefill\par\offinterlineskip\vskip1pt
    \textbf{#1#2}#3\offinterlineskip\hrulefill}
\DeclareCaptionStyle{algori}{singlelinecheck=off,format=algor,labelsep=space}
\usepackage{breqn}

\allowdisplaybreaks

\makeatletter
\let\cref@old@eq@setnumber\eq@setnumber
\def\eq@setnumber{
\cref@old@eq@setnumber
\cref@constructprefix{equation}{\cref@result}
\protected@xdef\cref@currentlabel{
[equation][\arabic{equation}][\cref@result]\p@equation\theequation}}
\makeatother

\title{Quantum adiabatic optimization without heuristics}
\author[1]{Michael Jarret \thanks{mjarret@pitp.ca}}
\affil[1]{\footnotesize\PITP}
\author[2,3,4]{Brad Lackey \thanks{bclackey@umd.edu}}
\affil[2]{\footnotesize Institute for Advanced Computer Studies, University of Maryland, College Park, MD 20742, USA}
\affil[3]{\footnotesize Departments of Computer Science and Mathematics, University of Maryland, College Park, MD 20742, USA}
\affil[4]{\footnotesize Quantum Architectures and Computation Group, Microsoft, Redmond, WA 98052, USA}
\author[5]{Aike Liu \thanks{aikeliu2@illinois.edu}}
\affil[5]{Department of Physics, University of Illinois Urbana-Champaign, Urbana, IL 61801, USA}
\author[6]{Kianna Wan \thanks{kianna@stanford.edu}}
\affil[6]{Stanford Institute for Theoretical Physics, Stanford University, Stanford, CA 94305, USA}
\date{\today}

\begin{document}

\maketitle

\begin{abstract}
Quantum adiabatic optimization (QAO) is performed using a time-dependent Hamiltonian $H(s)$ with spectral gap $\gamma(s)$. Assuming the existence of an oracle $\oracle$ such that $\gamma(s) = \Theta\left(\oracle(s)\right)$, we provide an algorithm that reliably performs QAO in time $\bigO{\gamma_\min^{-1}}$ with $\bigO{\log(\gamma_\min^{-1})}$ oracle queries, where $\gamma_\min \coloneqq \min_s \gamma(s)$. Our strategy is not heuristic and does not require guessing time parameters or annealing paths. Rather, our algorithm naturally produces an annealing path such that $\|dH/ds\| \approx \gamma(s)$ and chooses its own runtime to be as close as possible to optimal while promising convergence to the ground state. 
    
    We then demonstrate the feasibility of this approach in practice by explicitly constructing a gap oracle $\oracle$ for the problem of finding the minimum point $m \coloneqq \argmin_u W(u)$ of the cost function $W:\mathcal{V}\longrightarrow [0,1]$, restricting ourselves to computational basis measurements and driving Hamiltonian $H(0)=I - |\mathcal{V}|^{-1}\sum_{u,v \in \mathcal{V}}\ket{u}\bra{v}$. Requiring only that $W$ have a constant lower bound on its spectral gap and upper bound $\kappa$ on its spectral ratio, our QAO algorithm returns $m$ with probability $(1-\epsilon)(1-e^{-1/\epsilon})$ in time $\widetilde{\mathcal{O}}(\epsilon^{-1}[\sqrt{|\mathcal{V}|} + (\kappa-1)^{2/3}|\mathcal{V}|^{2/3}])$. This achieves a quantum advantage for all $\kappa$, and recovers Grover scaling up to logarithmic factors when $\kappa \approx 1$. We implement the algorithm as a subroutine in an optimization procedure that produces $m$ with exponentially small failure probability and expected runtime $\widetilde{\mathcal{O}}(\epsilon^{-1}[\sqrt{|\mathcal{V}|} + (\kappa-1)^{2/3}|\mathcal{V}|^{2/3}])$ even when $\kappa$ is not known beforehand. 
\end{abstract}
  \tableofcontents

\section{Introduction}\label{sec:introduction}

Despite both academic and industry interest in quantum adiabatic optimization (QAO), theoretical progress has moved slowly since the algorithm's introduction by Farhi et al. in 2000 \cite{farhi2000quantum}. Early on and in a major milestone, Roland and Cerf showed that QAO can achieve a quadratic improvement over classical unstructured search by exploiting complete knowledge of $\gamma(s)$, the ``spectral gap'' of the corresponding Hamiltonian which will be defined below \cite{roland2002quantum}. Nonetheless, taking advantage of $\gamma(s)$ is not generally possible; other than in very special circumstances, we are not privy to such strong information. To make matters worse, obtaining useful analytic bounds on the gap is a challenging analysis problem and generally infeasible. Thus, in lieu of rigor, focus has primarily been placed on numerical experiments, limited comparisons with known classical algorithms, and actual experiments using quantum annealing devices \cite{Lidar2018,Chapuis2017,denchev2016computational}. In this paper, we pursue a more general, rigorous framework for understanding QAO.

QAO is a restricted form of quantum annealing, a heuristic strategy that may exploit the power of quantum mechanics to solve optimization problems \cite{farhi2000quantum}. In general, to perform quantum annealing we begin by preparing a quantum state and then varying the system Hamiltonian with the hope that the resulting quantum state approximates the solution to an optimization problem \cite{kadowaki1998quantum}. QAO makes this strategy precise by starting with a known Hamiltonian $H(0)$ and ``slowly'' varying the system into $H(1)$, a system which has a ground state that solves the optimization problem of interest. By initially preparing the ground state $\phi_0(0)$ of $H(0)$ and performing a sufficiently slow variation over the family $H(s)$, we prepare a state close to the ground state of $H(s)$ for all $s \in [0,1]$ \cite{Jansen2006}. The timescale needed is usually related to the minimal spectral gap of the one-parameter family $H(s)$, or the minimum difference between the two lowest eigenvalues of $H(s)$. That is, if $H(s)$ has spectrum $\lambda_0(s) < \lambda_1(s) \leq \dots \leq \lambda_{N-1}(s)$ and \emph{spectral gap} $\gamma(s) = \lambda_1(s)-\lambda_0(s)$, then the difference in energy $\gamma_\min = \min_{s \in [0,1]} \gamma(s)$ specifies the timescale needed to guarantee convergence to the ground state of $H(s)$ for all $s \in [0,1]$.

Since we typically cannot determine $\gamma(s)$ analytically, the problem of how to choose the family $H(s)$--the adiabatic \emph{schedule}--is largely open. One common strategy is to search for schedules that work for most, instead of all, possible problems \cite{zeng2016schedule}. Unfortunately, the standard approach to optimizing a schedule requires a detailed analysis of the entire family \cite{Albash2018}. Additionally, current suggestions for optimally controlling an adiabatic evolution depend on minimizing an associated Lagrangian, which may require knowledge of the function to be optimized (up to permutation) or an approximation of the spectral gap of the entire family $H(s)$ \cite{avron2010optimal,reza,zulkowski,takahashi2018hamiltonian}. Other methods either abandon the adiabatic procedure altogether \cite{yang2017optimizing,takahashi2018hamiltonian} or alter the Hamiltonian \cite{crosson2014different,zhuang}. These methods all may fail for large classes of problems.

An alternative approach one might consider is an adaptive QAO protocol that is able to adjust the rate of change of the corresponding Hamiltonian in response to experimental measurements. An implicit assumption in any such strategy is that there exists some feedback from the physical system which we can exploit to adjust the rate of change $dH/ds$. Although in \cite{avron2010optimal} the authors show that dephasing determines the optimal rate of change $dH/ds$, this is not obviously an experimentally useful quantity and still often seems to be intimately related to the spectral gap. Thus, when we assume that we receive feedback from our system at a point $s_0$, we effectively claim that our experimentalist (or procedure) is capable of playing the role of an oracle for the spectral gap $\gamma(s_0)$. This gives rise to some natural questions: 
\begin{enumerate}
\item Assuming that there exists an oracle for the spectral gap, what runtime can a QAO procedure achieve and with how many queries? 
\item How well would such an oracle need to approximate the spectral gap? and 
\item Is the existence of such an oracle \emph{really} a crazy proposition?
\end{enumerate}

The first half of this paper addresses questions 1 and 2. In particular, we show that provided a Hamiltonian $H(s)$ (for $s\in[0,1]$) with spectral gap $\gamma(s)$ and an oracle $\Gamma(s)$ such that
\begin{enumerate}
    \item $\Gamma \leq \gamma$,
    \item the measure of the sublevel sets of $\Gamma(s)$ obeys the inequality $\mu \{s \vert \Gamma(s) \leq x \norm{H} \} \leq C x$ for some constant $C$ independent of the problem size and for all $x \in [0,1]$, and
    \item $\min_s \Gamma(s) = \Theta(\gamma_\min)$,
\end{enumerate}
then an adiabatic optimization algorithm can be completed in time $\bigO(\norm{H}/\gamma_\min)$ with $\bigO{\log \left(\norm{H}/\gamma_\min\right)}$ oracle calls. A bit more simply, if the oracle always produces a lower bound on the gap, local minima of $\Gamma(s)$ do not have a ``width'' that grows with the system size, and the oracle minimum is within a constant factor of the actual minimal spectral gap, then we can efficiently and adaptively perform QAO.

The second half of this paper addresses question 3. In particular, given the difficulty of estimating the spectral gap, one should be skeptical that such an oracle can exist. Nonetheless, we prove our intuitions false and provide an explicit algorithm that efficiently reproduces the behavior of the oracle for a particular driving Hamiltonian $H(0)$. That is, we explicitly construct a gap oracle $\Gamma$ for the model in which the driving Hamiltonian $H(0)$ is the Laplacian $L$ of a complete graph and the final Hamiltonian $H(1)$ is an unknown diagonal matrix $W$.\footnote{We impose some further constraints on $W$, but those constraints are as weak as possible while still guaranteeing that $H(s)$ does not have an exponentially small gap exponentially close to $s=1$.} With bounded probability, we can return the ground state of $W$ with runtime $\widetilde{\mathcal{O}}(\sqrt{V} + (\kappa-1)V^{2/3})$, where $\kappa$ is the ratio of the largest to second-smallest eigenvalue of $H(1)$. Modulo logarithmic terms, our adiabatic algorithm scales optimally when presented with a problem similar to Grover search, and we achieve a quantum advantage over classical search in every setting explored. 

Although Grover-type algorithms are not practical on current hardware, we restrict our oracle construction to computational-basis measurements. Thus, even though our complete graph architecture is unlikely to be realized in the near future, our approach is motivated by experiment and consistent with procedures that we presently know how to perform. Furthermore, additional mathematical study of broader classes of driving Hamiltonians is likely to advance our understanding of QAO generally. In \cref{sec:discussion}, we review theoretical problems that our construction may address as well as additional tools that we need to develop.

\subsection{Overview of methods}

Our approach explores a new version of the adiabatic algorithm that discretizes the adiabatic process \cite{jarret2018hamiltonian}. Discretization methods similar to those we explore already exist, but their utility depends upon a serendipitous choice of discretization points. In order to make a good choice, one needs to determine an appropriate adiabatic schedule in advance and may require additional machinery beyond knowledge of the initial Hamiltonian $H(0)$ and final Hamiltonian $H(1)$ \cite{Boixo2009, chiang2014improved}. In contrast, our approach discovers a near-optimal path dynamically during the algorithm \cite{avron2010optimal}.\footnote{Within the context of this paper, optimality is considered in relation to the spectral gap, not dephasing.} The runtime above then matches the best cases achievable by the approach of \cite{Boixo2009} without randomized evolutions, a need to guess at discretization points, or substantial knowledge of the cost function.

To construct the gap oracle for the complete graph, we first derive appropriate inequalities. In particular, in \cite{jarret2018hamiltonian}, one of the authors shows that for stoquastic Hamiltonians $H$, one can bound the spectral gap by the \emph{weighted Cheeger constant} $h$ \cite{Chung2000}. This constant depends only upon information contained in the driving Hamiltonian $H(0)$ and the instantaneous ground state $\phi(s)$ and, at least in principle, is a measurable quantity. (Whether $h$ is measurable in practice probably depends on $H(0)$.) Specifically, if $H(0)=L$ corresponds to a graph of maximum degree $d$,
\begin{equation}\label{eqn:old_cheeger}
    2h \geq \gamma(H) \geq \sqrt{h^2 +d^2}-d.
\end{equation}
If this function satisfies the conditions enumerate in \cref{sec:introduction} and $h$ can be efficiently computed, then we can take $\Gamma = \sqrt{h^2 +d^2}-d$ as our oracle. Furthermore, since $h$ gives both upper and lower bounds on the gap that are independent of the problem size, any other method for estimating the gap of a stoquastic Hamiltonian is equivalent to determining $h$. Given that $h$ depends on substantially less information than $\gamma$, it is probably a more natural object of study.

Although $h$ can be efficiently estimated for classical graphs \cite{Arora}, we do not understand the difficulty of estimating $h$ in our setting of a known driving Hamiltonian and unknown cost function. In fact, the classical and quantum problems are hard to compare. Unlike the classical case, once we choose a particular $H(0)$, we can assume that we know $h$ of $H(0)$ with no computational overhead. Since it fixes our architecture, our driving Hamiltonian should remain unchanged even as we attempt to solve different problems.

Despite the promise of $h$ and \cref{eqn:old_cheeger} for most architectures, when $H(0)$ corresponds to the complete graph, condition 3  of \cref{sec:introduction} ($\gamma_\min = \Theta(\min_s \Gamma(s))$) is violated. Thus, in \cref{sec:spectral_theory} we introduce a tighter Cheeger inequality specialized to the complete graph and show that, surprisingly, $h$ can be efficiently estimated. One should note that since driving Hamiltonians are fixed by annealing hardware, deriving hardware-motivated inequalities is probably needed to obtain desirable results.

The organization of the paper is as follows: In \cref{sec:baa}, we present the algorithm for general optimization problems. Then, in \cref{sec:general_analysis}, we prove the time and query complexity claimed above, conditioned on the existence of an oracle for the spectral gap. Next, in \cref{sec:baa_cg}, we consider the restricted problem of adiabatic optimization with a complete graph driving Hamiltonian $H(0)$ and provide an explicit construction that efficiently implements the oracle $\Gamma$. (This section is rather technical, and a roadmap of the arguments will be provided therein.) Finally, in \cref{sec:discussion}, we discuss how our results and approach might advance our understanding of other open questions in QAO.

\section{The Bashful Adiabatic Algorithm (BAA)}\label{sec:baa}

In this section, we present a general algorithm for adiabatic state preparation. 
Starting with the easily preparable ground state of an initial Hamiltonian $H_0$, suppose that the time-dependent interpolation Hamiltonian $H(t) = (1-s(t))H_0 + s(t)H_1$ is applied for a time $T$. Here, $H_1$ is a time-independent Hamiltonian whose ground state we wish to prepare, and $s(t)$ is some function of time, with $s(0) = 0$ and $s(T) = 1$. According to the quantum adiabatic theorem, the state at a time $t$ will remain close to the instantaneous ground state of $H(t)$ provided that this Hamiltonian varies sufficiently slowly in time. To achieve a desired accuracy, the maximum allowable evolution rate is determined by the spectral gap $\gamma(t)$ of the Hamiltonian: since a smaller gap indicates a higher probability of transition to a first excited state, the evolution must proceed slower when $\gamma(t)$ is small than when it is large. As demonstrated in \cite{roland2002quantum}, an optimal schedule $s(t)$---one that maximizes the evolution rate while ensuring that the adiabaticity condition is satisfied at each point in the evolution---can be found for Grover's search problem by analyzing an explicit expression for $\gamma(t)$.

In the general setting, however, the behaviour of the spectral gap as a function of time is not known \textit{a priori} (as it is for the Hamiltonian underlying the adiabatic algorithm for Grover search). Consequently, the analytic approach of \cite{roland2002quantum} cannot be used to design a schedule in advance. We propose instead an adaptive method in which the evolution is discretized into a series of ``checkpoints" $\{s_i\}_i$. At each $s_i$, we estimate the gap $\gamma(s_i)$ at that point and use the estimate to select the next checkpoint, $s_{i+1}$. Then, estimates of $\gamma(s_i)$ and $\gamma(s_{i+1})$ can be used in conjunction with Weyl's inequality to determine the appropriate evolution rate for the interval $(s_i, s_{i+1})$. Given a full profile $\{(s_i, \gamma_i)\}_i$, where $\gamma_i \approx \gamma(s_i)$, it is thus straightforward to develop an algorithm for adiabatic state preparation. The preparation procedure is formalized as \cref{alg:adiabatic} below and analyzed in \Cref{sec:adiabatic,sec:perturb}.

\begin{algorithm}[H]
 \caption{\label{alg:adiabatic}Adiabatic State Preparation}
 \begin{algorithmic}[1]
  \Require A time-dependent Hamiltonian $H(s)$ linear in $s \in [0,1]$, an upper bound $\lambda_{\mathrm{max}}$ on $\norm{H(s)}$, an array $\vec\gamma$, a universal constant $c_0 \in (0,1)$, a parameter $\epsilon >0$. 
 \Function{GenerateState}{$\vec\gamma$,$\delta s$,$\epsilon$} \Comment{Adiabatically generate a state according to \cref{thm:adiabatic3}}
    \State $(s_f,\gamma_f) = \Call{Last}{\vec\gamma}$ \Comment{Get last element of $\vec\gamma$}
    \If{$\delta s \neq 0$} \Comment{Extend schedule }
        \State $s_f \gets \min(s_f + \delta s, 1)$
        \State $\gamma_f \gets \gamma_f-4\delta s\lambda_\max$  \Comment{Lower-bound gap at new point using \cref{thm:weyl}}
        \State Append $(s_f, \gamma_f)$ to $\vec\gamma$
    \EndIf
   \State $\widetilde{P} \gets \ket{\phi(0)}\bra{\phi(0)}$ \Comment{Prepare the ground state of $H(0)$}
   \State $(s_0,\gamma_0)\gets \Call{Next}{\vec\gamma}$ \Comment{Get first element of list}
\While{$(s_1,\gamma_1)\gets \Call{Next}{\vec\gamma}$} \label{alg:adiabatic:loop} \Comment{Loop over provided schedule}
     \State $\gamma_\min \gets \dfrac{1}{2}(\gamma_{0} + \gamma_{1}) - 2(s_{1}-s_{0})\lambda_\max$ \label{alg1:gamma_min} \Comment{Lower bound on gap for $s \in [s_0,s_1]$, by \cref{intersect}}
    \State $H(\sigma) \gets \dfrac{1}{\lambda_\max}\left[(1-\sigma)H(s_{0}) + \sigma H(s_{1})\right]$ \Comment{Version of $H$ with norm $\bigO{1}$}
    \State $T \gets \left(c_0 + \dfrac{7}{4}c_0^2 \right)\dfrac{\lambda_\max}{\epsilon\gamma_\min}$ \label{alg1:time} \Comment{Set time using \cref{thm:adiabatic4}}
    \State $\widetilde{P}\gets U\widetilde{P}U^{\dagger}$ using \cref{eqn:schrodinger_unitary} with $H(s)\mapsto H(\sigma)$ and time $T$
    \State $(s_0,\gamma_0)\gets (s_1,\gamma_1)$
 \EndWhile 
 \State \Return $\widetilde{P}$
 \EndFunction
 \end{algorithmic}
\end{algorithm} 

We refer to the algorithm for adaptively determining a discretized evolution schedule as the Bashful Adiabatic Algorithm (BAA).\footnote{The name reflects the fact that, in each iteration, the algorithm ventures just one small step further than it did previously. Since estimating the gap may in general involve measuring the state, we subsequently reset the state to the ground state of $H(0)$, and evolve to $s_{i+1}$ on the next run, following the schedule that has been determined thus far. This is repeated until a checkpoint close to $1$ is reached.} We present BAA as \Cref{alg:BAA}. The checkpoints ultimately produced by BAA are conceptually similar to those proposed in the randomization method of \cite{Boixo2009}. Nonetheless, the method of \cite{Boixo2009} requires one to somehow choose checkpoints \emph{a priori} which may or may not be possible. Although the randomization method decreases dependency on the gap, it scales quadratically in the number of checkpoints and thus the efficiency of the adiabatic process depends on how well one chooses checkpoints. Unlike \cite{Boixo2009}, under relatively weak constraints, BAA automatically picks checkpoints such that a near-optimal scaling in the spectral gap is guaranteed \cite{avron2010optimal}. 

\begin{algorithm}[H]
 \caption{\label{alg:BAA}Bashful Adiabatic Algorithm}
 \begin{algorithmic}[1]
    \Require Time-independent Hamiltonians $H_0$ and $H_1$, an upper bound $\lambda_\max$ on $\max\{\norm{H_0}, \norm{H_1}\}$, the spectral gap $\gamma_0$ of $H_0$, a universal constant $c_0 \in (0,1)$, a parameter $\epsilon>0$.
 \Function{BAA}{\textproc{GetGap}}
 \State $s\gets 0$
 \State $\gamma \gets \gamma_0$
 \State $\vec\gamma  \gets [(0,\gamma)]$ \Comment{Initialize the array for $s=0$}
 \While{$s<1$} \label{alg:BAA_while}\Comment{Initialize the adiabatic schedule}
  \State \label{alg:BAA_oracle_ds} $\delta s \gets \min\left(\dfrac{c_0\gamma}{4\lambda_\max},1-s\right)$ \Comment{Step size such that $\abs*{\gamma(s+\delta s)-\gamma(s)} \leq c_0\gamma(s)$, by \cref{thm:delta_s}}
 \State $\gamma \gets \Call{GetGap}{H_0, H_1, s,\delta s,\gamma}$ \Comment{Query the oracle for the next gap}
 \State $s \gets s + \delta s$
 \State \label{alg:BAA_endwhile}Append $(s,\gamma)$ to $\vec{\gamma}$
 \EndWhile
 \State \Return $\Call{GenerateState}{\vec\gamma,0,\epsilon}$ \Comment{Return the adiabatically prepared ground state of $H(1)$}
 \EndFunction
 \end{algorithmic}
\end{algorithm}

It is important to note that in \Cref{alg:BAA}, we treat the routine $\textsc{GetGap}$, which estimates the spectral gap, as a black box, and proceed to analyze the query complexity of \Cref{alg:BAA} in \cref{sec:oracle-calls}. Of course, the assumption of having access to such a gap oracle seems highly unrealistic in most settings, but we substantiate our approach in \cref{sec:the_oracle} by providing an \emph{explicit} construction for $\textsc{GetGap}$ for the case of unstructured optimization. The problems that we explore can be seen as similar to standard Grover search, but without a similarly powerful oracle. We prove that our construction is efficient and reliable in \cref{sec:statistics}. Our combined analysis therefore demonstrates that, at least in some circumstances, QAO can be performed efficiently and reliably, without assuming access to black boxes.

\section{Analysis of BAA}\label{sec:general_analysis}
\subsection{Variations of the quantum adiabatic theorem}\label{sec:adiabatic}

Most quantitative versions of the quantum adiabatic theorem appear in a form similar to that of Theorem \ref{thm:adiabatic} below. For a total evolution time $T$, we define a scaled time parameter $s \coloneqq t/T$, and consider the unitary $\bar{U}(t)$ that evolves the system governed by Hamiltonian $H(t/T)$. The Schr\"odinger equation reads 
\begin{equation*}
      \imath \frac{d \bar{U}(t)}{dt} = H(t/T)\bar{U}(t), \qquad
      \bar{U}(0) = I,
\end{equation*}
or, equivalently, 
\begin{equation} \label{eqn:schrodinger_unitary}
\imath\frac{d U(s)}{ds} = TH(s)U(s), \qquad U(0) = I,
\end{equation}
with $U(s) = \bar{U}(t)$.

In \cite{Jansen2006}, the authors provide quantitative bounds on the deviation of an adiabatically prepared state from the instantaneous ground state. We restate one of their results, restricting to Hamiltonians with non-degenerate ground states.

\begin{thm}[{\cite[Theorem 3]{Jansen2006}}]\label{thm:adiabatic} 
Suppose that $H(s)$ is a self-adjoint operator on an $N$-dimensional Hilbert space, with eigenvalues $\lambda_0(s) < \lambda_1(s) \leq \dots \leq \lambda_{N-1}(s)$. For $s\in [0,1]$, let $P(s)$ denote the projector onto the ground state of $H(s)$, so that $H(s)P(s) = \lambda_0(s) P(s)$, and define $\widetilde{P}(s) \coloneqq U(s)P(0)U(s)^\dagger$, where $U$ satisfies \cref{eqn:schrodinger_unitary}. Then
\begin{equation*}
    \norm{\widetilde{P}(s) - P(s)} \leq A(s),
\end{equation*}
where
\begin{equation*}
    A(s) \leq \frac{1}{T}\left[\frac{ \|\dot{H}(0)\|}{\gamma(0)^2} + \frac{ \norm{\dot{H}(s)}}{\gamma(s)^2} + \int_0^s ds' \left(7\frac{ \norm{\dot{H}(s')}^2}{\gamma(s')^3} + \frac{ \|\ddot{H}(s')\|}{\gamma(s')^2}   \right) \right]
\end{equation*}
with $\gamma(s) \coloneqq \lambda_1(s) - \lambda_0(s)$ and
\begin{equation*} \dot{H}(a) = \restr{\frac{dH(s)}{ds}}{s=a}, \qquad \ddot{H}(a) = \restr{\frac{d^2H(s)}{ds^2}}{s=a}.\end{equation*}
\end{thm} 
A simple restriction of \cref{thm:adiabatic} yields the following. 
\begin{prop}\label{thm:adiabatic2}
Suppose all quantities are defined as in \cref{thm:adiabatic}. If $\norm{\dot{H}(s)} \leq c_0 \gamma(s)/2$ and $\norm{\ddot{H}(s)} \leq c_1 \gamma(s)$ for all $s\in[0,1]$, then 
\begin{equation*}
    A(s) \leq \frac{1}{T}\left(\frac{c_0 + 7c_0^2/4 + c_1}{\gamma_\min}\right), 
\end{equation*}
where $\gamma_\min\coloneqq \min_{s\in [0,1]}\gamma(s)$
\end{prop} 
Hence, if both the first and the second derivatives of $H(s)$ are upper-bounded in terms of the spectral gap $\gamma(s)$ and constants $c_0, c_1$ independent of $\gamma_\min$, then the adiabatic theorem guarantees that scaling $T$ as $\gamma_\min^{-1}$ can reduce $A$ to within any constant error. Although the conditions of \cref{thm:adiabatic2} are not satisfied for $c_0,c_1$ by most time-dependent Hamiltonians, we can perform a computationally equivalent process by chaining together a sequence of interpolating Hamiltonians.

\begin{thm}\label{thm:adiabatic3}
For a given self-adjoint operator $H(s)$ on a finite-dimensional space, define an associated operator $H_L(s)$ piecewise linear in $s$ as
\begin{equation*}
    H_L(s) = \begin{dcases} \frac{T_0}{s_1}H(0), \qquad &s =0 \\
    \sum_{i=0}^{q-1}\frac{T_i}{s_{i+1}-s_i} \left[\left(1-\frac{s-s_i}{s_{i+1}-s_i}\right)H(s_i) + \frac{s-s_i}{s_{i+1}-s_i}H(s_{i+1})\right]\mathbbm{1}_{(s_i,s_{i+1}]}(s), \qquad &s\in(0,1]
    \end{dcases}
\end{equation*}
for $0=s_0 < s_1 < \dots < s_{q} = 1$ and $T_0, T_1, \dots, T_{q-1} > 0$. For all $s\in[0,1]$, let $P_L(s)$ denote the projector onto the ground state of $H_L(s)$, and define $\widetilde{P}_L(s) \coloneqq U_L(s)P_L(0)U_L(s)^\dagger$, where $U_L$ is the unitary satisfying
\begin{equation*}  
    \imath \frac{d U_L(s)}{ds} = H_L(s)U_L(s), \qquad 
    U_L(0)=I.
\end{equation*}
Then 
\begin{equation} \label{induction}
    \norm{\widetilde{P}_L(s_j)-{P}_L(s_j)} \leq \sum_{i=0}^{j-1} A_i
\end{equation}
for all $j \le q$, where $A_i$ is any upper bound that results from applying \cref{thm:adiabatic} to the Hamiltonian $(1-\sigma_i)H(s_i)+\sigma_i H(s_{i+1})$ over the interval $\sigma_i \in [0,1]$ with total evolution time $T_i$. 
\end{thm}
\begin{proof} For $s \in (s_i, s_{i+1}]$, $U_L$ can be written as
    \begin{equation*}
        U_L(s) = U_{i}(s)U_{i-1}(s_{i})\cdots U_{0}(s_{1})
    \end{equation*} where for each $i = 0,1,\dots,q-1$, $U_i$ satisfies 
    \begin{equation}\label{eqn:adiabatic_piecewise}
        \imath \frac{d U_i(s)}{ds} = \frac{T_i}{s_{i+1}-s_i} \left[\left(1-\frac{s-s_i}{s_{i+1}-s_i}\right)H(s_i) + \frac{s-s_i}{s_{i+1}-s_i}H(s_{i+1})\right]U_i(s), \qquad U_i(s_i) = I
    \end{equation}
    Making a change of variables 
    \begin{equation*}
        \sigma_i = \frac{s - s_i}{s_{i+1}-s_i}
    \end{equation*}
    and defining \begin{equation*}
    H_i(\sigma_i) = (1-\sigma_i)H(s_i) + \sigma_i H(s_{i+1})
    \end{equation*}
    for $\sigma_i \in [0,1]$,
    we see that \cref{eqn:adiabatic_piecewise} is equivalent to 
    \begin{equation*}
        \imath  \frac{dU_i(s)}{d\sigma_i} = T_i H_i(\sigma_i)U_i(s), \qquad U_i(\sigma_i = 0) = I,
    \end{equation*}    
    so \cref{thm:adiabatic} applies to each $H_i$ over the interval $\sigma_i \in [0,1]$ (i.e., $s \in [s_{i}, s_{i+1}]$), with $T \to T_i$ and projector $P_L(s_i)$ onto the ground state at $\sigma_i = 0$. We prove Eq.\ \eqref{induction} by induction. The base case holds since $\tilde{P}_L(0) = P_L(0)$. Assume Eq.\ \eqref{induction} is true for some integer $j\geq 1$. Then,
    \[     \norm{\widetilde{P}_L(s_{j+1})-P_L(s_{j+1})} 
    =  \norm{U_L(s_{j+1})P_L(0)U_L(s_{j+1})^\dagger - P_L(s_{j+1})}
        = \norm{U_j(s_{j+1})U_L(s_j)P_L(0)U_L(s_j)^\dagger U_j(s_{j+1})^\dagger - P_L(s_{j+1})}
        = \norm{U_j(s_{j+1})\widetilde{P}_L(s_{j})U_j(s_{j+1})^\dagger - P_L(s_{j+1})}
        = \norm{U_j(s_{j+1})(\widetilde{P}_L(s_j) - P_L(s_j))U_j(s_{j+1})^\dagger + U_j(s_{j+1})P_L(s_j)U_j(s_{j+1})^\dagger - P_L(s_{j+1})} 
        \leq \norm{\widetilde{P}_L(s_j) - P_L(s_j)} + \norm{U_j(s_{j+1})P_L(s_j)U_j(s_{j+1})^\dagger - P_L(s_{j+1})}
        \leq \sum_{i=0}^{j-1}A_i + A_j = \sum_{i=0}^{j}A_i,
    \]
    where the second inequality follows from applying the inductive hypothesis to the first term and \cref{thm:adiabatic} to the second. Eq.\ \eqref{induction} is thus true for all $j$.
\end{proof} 
\begin{cor}\label{thm:adiabatic4}
    Suppose that $H(s)$ is linear in $s$. Let all quantities be defined as in \cref{thm:adiabatic3}, and let $\gamma_{i,\min} \coloneqq \min_{s\in [s_i, s_{i+1}]}\gamma(s)$, where $\gamma(s)$ denotes the spectral gap of $H(s)$. If for all $i$,
    \begin{equation*} s_{i+1}-s_i \leq \frac{c_0 \gamma_{i,\min}}{2 \norm{\dot{H}}}, \qquad  T_i \geq \frac{q(c_0 + 7c_0^2/4)}{\epsilon\gamma_{i,\min}} \end{equation*} for some $c_0 >0$ and constant $\epsilon >0$,
    then
    \begin{equation*}
        \norm{\widetilde{P}_L(1)-P_L(1)} \leq \epsilon.
    \end{equation*}
\end{cor}
\begin{proof} By \cref{thm:adiabatic3},
    \begin{equation*}
        \norm{\widetilde{P}_L(1)-P_L(1)} = \norm{\widetilde{P}_L(s_{q}) - P_L(s_q)} \leq \sum_{i=0}^{q-1}A_i.
    \end{equation*}
Since $\norm{d^2H_i(\sigma_i)/d\sigma_i^2} = 0$ and
\begin{equation*} \norm*{\frac{dH_i(\sigma_i)}{d\sigma_i}} = \norm{H(s_{i+1}) - H(s_i)} = \norm{(s_{i+1} - s_i)\dot{H}} \leq \frac{c_0\gamma_{i,\min}}{2} \end{equation*}
by our assumption on the step sizes $s_{i+1} - s_i$, and, \cref{thm:adiabatic2} applies to each $A_i$ independently, giving
\begin{equation*}
    \sum_{i=0}^{q-1} A_i \leq \sum_{i=0}^{q-1}\frac{1}{T_i}\left(\frac{c_0 + 7c_0^2/4}{\gamma_{i,\min}}\right) \leq \sum_{i=0}^{q-1}\frac{\epsilon}{q} = \epsilon,
\end{equation*}
where the second inequality follows from our assumption on $T_i$. 
\end{proof} 
Typically, a weaker form of \cref{thm:adiabatic} is used, with
\begin{align} 
    A(s) &\leq \frac{1}{T}\left[\frac{2\max\limits_{s'\in\{0,s\}}\norm{\dot{H}(s')}}{\left(\min_{s\in [0,1]}\gamma(s)\right)^2} + \int_0^1 ds' \left(7\frac{\norm{\dot{H}(s')}^2}{\gamma(s')^3} + \frac{ \norm{\ddot{H}(s')}}{\gamma(s')^2}   \right) \right] \eqqcolon A, \label{eqn:adiabatic_standard}
\end{align}
for all $s \in [0,1]$. If $H(s)$ is linear in $s$, the upper bound simplifies to
\begin{equation} \label{A}
    A = \frac{1}{T}\left[\frac{2\norm{\dot{H}}}{\left(\min_{s\in [0,1]}\gamma(s)\right)^2} + 7\int_0^1ds\frac{\norm{\dot{H}}^2}{\gamma(s)^3} \right].
\end{equation} 
Since each $H_i(\sigma_i)$ considered in \cref{thm:adiabatic3} is linear in $\sigma_i$, we see that we can take the upper bounds $A_i$ in \cref{thm:adiabatic3} to be
\begin{equation} \label{A_i}
    A_i = \frac{1}{T_i}\left[\frac{2}{\left(\min_{\sigma_i \in [0,1]}\gamma_i(\sigma_i)\right)^2}\norm*{\frac{dH_i(\sigma_i)}{d\sigma_i}} + 7 \int_0^1 \frac{d\sigma_i}{\gamma_i(\sigma_i)^3}\norm*{\frac{dH_i(\sigma_i)}{d\sigma_i}}^2\right],
\end{equation}
where $\gamma_i(\sigma_i) \coloneqq \gamma(\sigma_i(s_{i+1}-s_i) + s_i)$ is the spectral gap of $H_i(\sigma_i)$. 

The following theorem demonstrates that if we take a prescribed linear evolution over a time $T$ and divide it into a sequence of $k$ linear evolutions, each taking time $T_i\geq T$ for $i \in \llbracket k-1 \rrbracket$, the total divergence of the state prepared by \cref{alg:adiabatic} from the ground state does not scale with $k$. One can actually show that the divergence decreases, but doing so requires working directly with the proof of \cref{thm:adiabatic}, which is beyond the scope of this paper. 

\begin{thm}\label{thm:adiabatic5}
    Suppose that $H(s)$ is linear in $s$, and let all quantities be defined as in Theorems \ref{thm:adiabatic} and \ref{thm:adiabatic3}. If $T_i \geq T$ for all $i = 0,1,\dots, q-1$, then 
        \begin{equation*}
           \sum_{i=1}^{q-1} A_i \leq A,
        \end{equation*}
    with $A$ and $A_i$ as defined in Eqs.\ \eqref{A} and \eqref{A_i}.
\end{thm}
\begin{proof}
    Since $H(s)$ is linear in $s$, $
        \norm{\widetilde{P}(1) - {P}(1)} \leq A$,
    where $A$ is defined as in \cref{A}, and
    letting $\delta s_i \coloneqq s_{i+1} - s_i$, we have
    \begin{equation*} \norm*{\frac{dH_i(\sigma_i)}{d\sigma_i}} = \delta s_i \norm*{\frac{dH(s)}{ds}} \end{equation*}
    for all $i = 0,1,\dots, q-1$, where $\norm{dH(s)/ds}$ is constant.
Then, \cref{A} can be written
\begin{dgroup*}
\[
    \hiderel{A\;} = \frac{1}{T}\left[\frac{2}{\left(\min_{s\in [0,1]}\gamma(s)\right)^2}\sum_{i=0}^{q-1}\delta s_i \norm*{\frac{dH(s)}{ds}} + 7\sum_{i=0}^{q-1}\int_{s_i}^{s_{i+1}}\frac{ds}{\gamma(s)^3} \norm*{\frac{dH(s)}{ds}}^2\right]
\]\[
    = \sum_{i=0}^{q-1}\frac{1}{T}\left[\frac{2}{\left(\min_{s\in [0,1]}\gamma(s)\right)^2}\norm*{\frac{dH_i(\sigma_i)}{d\sigma_i}} + 7\int_{s_i}^{s_{i+1}}\frac{ds}{{\delta s_i}^2\gamma(s)^3}\norm*{\frac{dH_i(\sigma_i)}{d\sigma_i}}^2\right]
\]
\[
    = \sum_{i=0}^{q-1}\frac{1}{T}\left[\frac{2}{\left(\min_{s\in [0,1]}\gamma(s)\right)^2}\norm*{\frac{dH_i(\sigma_i)}{d\sigma_i}} + \frac{7}{{\delta s_i}}\int_0^1 \frac{d\sigma_i}{\gamma_i(\sigma_i)^3}\norm*{\frac{dH_i(\sigma_i)}{d\sigma_i}}^2\right]
\]\[ 
    \geq \sum_{i=0}^{q-1}\frac{1}{T}\left[\frac{2}{\left(\min_{s\in [s_i,s_{i+1}]}\gamma(s)\right)^2}\norm*{\frac{dH_i(\sigma_i)}{d\sigma_i}} + 7\int_0^1 \frac{d\sigma_i}{\gamma_i(\sigma_i)^3}\norm*{\frac{dH_i(\sigma_i)}{d\sigma_i}}^2\right].
\]
\begin{dsuspend}
\\Now, since $T_i \geq T$ for all $i$, we have
\end{dsuspend}
\[
    \hiderel{A\;} \geq { \sum_{i=0}^{q-1}\frac{1}{T_i}\left[\frac{2}{\left(\min_{\sigma_i\in [0,1]}\gamma_i(\sigma_i)\right)^2}\norm*{\frac{dH_i(\sigma_i)}{d\sigma_i}} + 7\int_0^1 \frac{d\sigma_i}{\gamma_i(\sigma_i)^3}\norm*{\frac{dH_i(\sigma_i)}{d\sigma_i}}^2\right] = \sum_{i=0}^{q-1} A_i}.
\]
\end{dgroup*}
\end{proof} 

\subsection{Perturbation bounds}\label{sec:perturb}
Since the bound of \cref{thm:adiabatic3} is additive in each linear component of the adiabatic interpolation, to obtain low error we require that each $A_j$ of \cref{thm:adiabatic3} satisfy \cref{thm:adiabatic2} with $c_0,c_1=\bigO{1}$. Hence, we determine an appropriate set $\{s_i\}_i$ such that these constants are known when \cref{thm:adiabatic2} is applied to each interval $[s_i,s_{i+1}]$. Weyl's inequality allows us to bound the change in the spectral gap in terms of the step size $s_{i+1}-s_i$.
\begin{prop}\label{thm:weyl}
    Suppose that $H(s)$ is linear in $s$ and let $\gamma(s)$ denote the spectral gap of $H(s)$. Then, for $\delta s > 0$,
    \begin{equation*}
        \abs*{\gamma(s+\delta s) - \gamma(s)} \leq  2\delta s\norm{\dot{H}}.   
    \end{equation*}
\end{prop}
\begin{proof} Since $H(s)$ is linear is $s$, 
$ H(s+\delta s) = H(s) + \delta s\dot{H}$.
It follows from a straightforward application of Weyl's inequality to the eigenvalues of $H(s + \delta s)$ that
\begin{equation*} \gamma(s) - 2\delta s\norm{\dot{H}} \leq \gamma(s + \delta s) \leq \gamma(s) + 2\delta s\norm{\dot{H}}. \end{equation*}
\end{proof}
Note that since \Cref{alg:BAA} constructs a schedule for the linear interpolation $H(s) = (1-s)H_0 + sH_1$, $\|\dot{H}\|$ is upper-bounded by $\|H_1 - H_0\| \leq 2\max\{\|H_0\|, \|H_1\|\} \leq 2\lambda_\max$. 
In Line \ref{alg:BAA_oracle_ds} of \Cref{alg:BAA}, we choose at each $s_i$ the largest subsequent step size $\delta s = s_{i+1}-s_i$ for which the change $\abs{\gamma(s_{i+1} )-\gamma(s_i)}$ is guaranteed by \cref{thm:weyl} to be small relative to $\gamma(s_i)$. This is made precise by the following proposition.
\begin{prop}\label{thm:delta_s}
    Suppose that $H(s)$ is linear in $s$ and let $\gamma(s)$ denote the spectral gap of $H(s)$. If $0 < \delta s \leq c_0\gamma(s)/2\norm{\dot{H}}$ for some $c_0 >0$, 
    \begin{equation*}\abs*{\gamma(s+ \delta s)-\gamma(s)} \leq c_0\gamma(s).\end{equation*}
\end{prop}
\begin{proof}
    This follows immediately from \cref{thm:weyl}. 
\end{proof}

\cref{thm:weyl} can also be used to find a lower bound on the gap in an interval $[s_i, s_{i+1}]$ in terms of $\gamma(s_i)$ and $\gamma(s_{i+1})$. Combining $\gamma(s) \geq \gamma(s_i) - 2(s-s_i)\|\dot{H}\|$ and $\gamma(s) \geq \gamma(s_{i+1}) - 2(s_{i+1}-s)\|\dot{H}\|$ for $s \in [s_i, s_{i+1}]$ and recalling that $\|\dot{H}\| \leq 2\lambda_\max$, we have
\begin{equation}\label{intersect} \gamma_{i,\min} \coloneqq \min_{s \in [s_i, s_{i+1}]}\gamma(s) \geq  \frac{1}{2}(\gamma(s_i) + \gamma(s_{i+1})) - 2(s_{i+1} - s_i)\lambda_\max. \end{equation} 
Given estimates of $\gamma(s_i)$ and $\gamma(s_{i+1})$, Line \ref{alg1:gamma_min} of \cref{alg:adiabatic} applies this bound to approximate $\gamma_{i,\min}$.  
\subsection{Query complexity}\label{sec:oracle-calls}

The following theorem shows that for $\textsc{GetGap}(s, \delta s, \cdot) = \Theta(\gamma(s+\delta s))$, the query complexity of \cref{alg:BAA} is logarithmic in $(\min_{s\in[0,1]}\gamma(s))^{-1}$ under fairly general assumptions about the spectral gap $\gamma(s)$. We consider sublevel sets of the form $I_k \coloneqq \{s \in [0,1] \:|\: \textsc{GetGap}(s,0,\cdot)\leq \lambda_\max/{2^k}\}$ and imposed two simple constraints. First, we require that the ``width" (in $s$) of local minima in the gap is not too large relative to the size of these minima.  
This condition is imposed as an upper bound on the measure $\mu(I_k)$ of each sublevel set $I_k$. Secondly, we assume that the gap does not oscillate too wildly (as a function of $s$), as quantified by the minimum number of intervals $I_k^{(l)}$ such that $I_k = \cup_l I_k^{(l)}$.

\begin{thm}\label{thm:queries_redux}\label{thm:queries}
Let all quantities be defined as in \cref{alg:BAA}, and let $\oracle(s) := \text{\textproc{GetGap}}(s,0,\cdot)$. Define the sublevel sets $I_k \coloneqq \{ s \in [0,1] \:|\: \oracle(s) \leq {\lambda_\max}/2^{k}\}$. If for all $k \in \mathbb{R}_{\geq 0}$, $\mu(I_k) \leq C/2^{k}$, where  $C$ is a constant independent of the problem size, and $I_k$ can be written as the union of $R$ intervals, then \cref{alg:BAA} makes $\bigO{R\log\left({\lambda_\max}/{\Gamma_\min}\right)}$ queries to \textproc{GetGap}, where $\Gamma_\min \coloneqq \min_{s \in [0,1]}\Gamma(s)$.
\end{thm}
\begin{proof}
Let $0=s_0 < s_1 < \dots < s_{q} = 1$ be the sequence of checkpoints determined by \cref{alg:BAA}. At each $s= s_i$, \cref{alg:BAA} chooses the subsequent point as
\begin{equation} \label{eq:ds}
    s_{i+1} = s_i + \frac{c_0\oracle(s_i)}{4\lambda_\max}, 
\end{equation}
(unless $c_0\Gamma(s_i)/4\lambda_\max > 1-s_i$, in which case it simply sets $s_{i+1} = 1$ and makes the final step in the evolution).  Write $I_k\setminus I_{k+1} = \bigcup_j J_j$, where the $J_j$ are disjoint intervals. Since $I_k$ and $I_{k+1}$ can each be written as the union of $R$ intervals, $I_k\setminus I_{k+1}$ can be written as the union of no more than $2R$ intervals. Define the set $S_k \coloneqq \{s_i \:|\: s_i \in I_k\setminus I_{k+1}\}$. Then, for every $s_i \in S_k$, we have $s_{i+1} - s_i = {c_0\oracle(s_i)}/{4\lambda_\max} > {c_0}/(4\cdot 2^{k+1})$, whence
\[
    \mu(J_j) \geq \sum_{[s_i,s_{i+1}] \subseteq J_j}(s_{i+1}-s_i) 
    = \sum_{\substack{s_i \in  J_j \cap S_k \\ s_{i+1} \in J_j}}(s_{i+1}-s_i) 
    \geq \frac{c_0}{4 \cdot 2^{k+1}}\left(\abs*{S_k \cap J_j}-1\right).
\]
or $\abs*{S_k \cap J_j} \leq 8\cdot 2^k\mu(J_j)/c_0 + 1$, where in the last line we use the fact that for any interval $J_j$, there exists at most one $s_i \in J_j$ for which $s_{i+1} \notin J_j$. Hence,
\begin{align}
    \abs*{S_k} &= 
    \sum_j\abs*{S_k \cap J_j} \nonumber\\\nonumber
    &\leq \sum_j\left(\frac{8\cdot 2^k\mu(J_j)}{c_0} + 1 \right) \\\nonumber
    &\leq \frac{8\cdot 2^k\mu(I_k)}{c_0} + 2R \\
    &\leq \frac{8C}{c_0} + 2R \label{Sk_bound}
\end{align}
for any $k$. Noting that the $S_k$ are disjoint and that $S_k = \emptyset$ for $k > \floor{\log(\lambda_\max/\Gamma_\min)}$, the total number of checkpoints $s_i$ is therefore
\begin{align}
    q +1 &= \sum_{k=0}^{\infty} |S_k| \nonumber \\
    &= \sum_{k=0}^{\floor*{\log(\lambda_\max/\Gamma_\min)}}|S_k| \nonumber\\ &\leq \sum_{k=0}^{\floor*{\log(\lambda_\max/\Gamma_\min)}}\left(\frac{8C}{c_0}+2R\right) \nonumber\\
    &=\left(\floor*{\log\left(\frac{\lambda_\max}{\Gamma_\min}\right)}+1\right)\left(\frac{8C}{c_0} + 2R\right). \nonumber
\end{align}
Since \cref{alg:BAA} makes one query to $\textsc{GetGap}$ at each $s_i < 1$, it follows that $q =\mathcal{O}(R\log({\lambda_\max}/\Gamma_\min))$ queries are made in total. \end{proof}

In particular, \cref{thm:queries} implies that if $\Gamma(s) = \Theta(\gamma(s))$ and $R = \mathcal{O}(1)$, then $\mathcal{O}(\log(\lambda_\max/\Gamma_\min))$ queries are required.

Recall that the final step of \cref{alg:BAA} is a call to  \textproc{GenerateState} (\cref{alg:adiabatic}), with the profile $\vec{\gamma}=\{(s_i,\gamma_i)\}_i$ determined by \cref{alg:BAA} as input. The following theorem bounds the total runtime of $\Call{GenerateState}$ using the schedule produced by \cref{alg:BAA}.

\begin{thm}\label{thm:adiabatic_prep}
Let all quantities be defined as in \cref{thm:queries_redux}. Under the conditions of \cref{thm:queries_redux}, if $\vec{\gamma}$ is the array constructed by \cref{alg:BAA}, then $\Call{GenerateState}{\vec\gamma,0,\epsilon}$ takes time $\bigO{R\epsilon^{-1}{\lambda_\max}/{\Gamma_\min}}$. 
\end{thm}
\begin{proof} 
When \cref{alg:BAA} calls \textsc{GenerateState} (\cref{alg:adiabatic}), the evolution time $T_i$ over the interval $[s_i, s_{i+1}]$ is set in Line \ref{alg1:time} of \cref{alg:adiabatic} as 
\[ T_i = \left(c_0 + \frac{7}{4}c_0^2\right)\frac{\lambda_\max}{\epsilon\gamma_{i,\min}}, \]
where
\[ \gamma_\min = \frac{1}{2}(\Gamma(s_i) + \Gamma(s_{i+1})) -2(s_{i+1} - s_i)\lambda_\max.\]
By \cref{eq:ds}, the step sizes are chosen by BAA such that $s_{i+1} - s_i = c_0\Gamma(s_i)/4\lambda_\max$, so 
\begin{equation*} \gamma_{i,\min} = \frac{1}{2}(\Gamma(s_i) + \Gamma(s_{i+1}) - c_0\Gamma(s_i)) > \frac{1}{2}\Gamma(s_{i+1})\end{equation*}
since $c_0 \in (0,1)$, and hence
\[ T_i < 2\left(c_0 + \frac{7}{4}c_0^2\right)\frac{\lambda_\max}{\epsilon\Gamma(s_{i+1})}. \] Thus,
the total evolution time $T = \sum_{i=0}^{q-1}T_i$ is bounded as
\begin{align*}
    T
    &< 2\epsilon^{-1}\left(c_0 + \frac{7}{4}c_0^2\right)\sum_{i=0}^{q-1}\frac{\lambda_\max}{\Gamma(s_{i+1})} \\
    &< 2\epsilon^{-1}\left(c_0 + \frac{7}{4}c_0^2\right) \sum_{k=0}^{\infty}\sum_{s_i\in S_k}\frac{\lambda_\max}{\Gamma(s_i)} \\
    &< 2\epsilon^{-1}\left(c_0 + \frac{7}{4}c_0^2\right)\sum_{k=0}^{\infty}\sum_{s_i\in S_k}2^{k+1} \\
    &= 2\epsilon^{-1}\left(c_0 + \frac{7}{4}c_0^2\right)\sum_{k=0}^{\lfloor \log(\lambda_\max/\Gamma_\min)\rfloor}|S_k|2^{k+1} \\
    &\leq 2\epsilon^{-1}\left(c_0 + \frac{7}{4}c_0^2\right)\left(\frac{8C}{c_0} + 2R\right)\sum_{k=0}^{\lfloor \log(\lambda_\max/\Gamma_\min)\rfloor}2^{k+1} \\
    &< 8\epsilon^{-1}\left(c_0 + \frac{7}{4}c_0^2\right)\left(\frac{8C}{c_0} + 2R\right)\left(\frac{\lambda_\max}{\Gamma_\min}\right)\\
    &=\mathcal{O}\left(\epsilon^{-1}R\frac{\lambda_\max}{\Gamma_\min}\right),
\end{align*}
where $S_k$ is defined as in the proof of \cref{thm:queries} and the second last inequality follows from \cref{Sk_bound}.
\end{proof}
The next theorem demonstrates that \cref{alg:BAA} successfully produces a state within $\bigO{\epsilon}$ of the ground state of $H(1)$ in time $\bigO{\epsilon^{-1}\lambda_\max/\gamma_\min}$.
\begin{thm} \label{thm:adiabatic_final}
Under the conditions of \cref{thm:adiabatic_prep}, if $\Gamma \leq \gamma$, $\Gamma_\min = \Theta(\gamma_\min)$, $R = \bigO{1}$, and $P(1)$ is the projector onto the ground state of $H(1)$, then \cref{alg:BAA} produces a projector $\widetilde{P}(1)$ such that $\norm{\widetilde{P}(1) - P(1)} = \bigO{\epsilon}$ in time $\bigO{\epsilon^{-1}\lambda_\max/\gamma_\min}$.
\end{thm}

\begin{proof}
    This is an immediate consequence of \cref{thm:adiabatic_prep} and the fact \cref{alg:adiabatic}, Line \ref{alg1:time} sets  $T_i$ to be larger than required by \cref{thm:adiabatic4}. That is, for each $T_i$ in \cref{thm:adiabatic_prep},
        \begin{equation*}
            T_i \geq \left(c_0 + \frac{7}{4}c_0^2\right)\frac{\lambda_\max}{\epsilon \min_{s \in [s_i,s_{i+1}]}\gamma(s)}.
        \end{equation*}
    Thus, \cref{thm:adiabatic4} yields the error bound $\norm{\widetilde{P}(1) - P(1)} = \bigO{\epsilon}$. Applying the fact that $\Gamma_\min = \Theta(\gamma_\min)$ and $R = \bigO{1}$ to \cref{thm:adiabatic_prep} yields the runtime bound $\bigO{\epsilon^{-1}\lambda_\max/\gamma_\min}$.
\end{proof}

\section{BAA on the complete graph}\label{sec:baa_cg}

In the following sections, we justify our assumption of the existence of the gap oracle \textsc{GetGap} used by \cref{alg:BAA} by explicitly constructing \textsc{GetGap} for a more specific---though still rather general---class of optimization problems. We consider the setting where the initial Hamiltonian $H_0$ is the combinatorial Laplacian $L$ of the complete graph on a vertex set $\mathcal{V}$ of size $V$ and the final Hamiltonian $H_1$ is a diagonal matrix $W$ whose entries $W_{uu} \eqqcolon W_u$ are the values of an unknown cost function on $\mathcal{V}$. \cref{fig:problems} shows how difficult performing adiabatic optimization in this simple scenario can be. Even in the much more restricted case where $W$ is proportional to the Grover cost function, i.e., $W =\mathrm{diag} (0,C,C,\dots,C)$ for some $C > 0$, the gap profile can vary significantly for different values of $C$, and so viable annealing schedules cannot be guessed reliably. In particular, the position of the minimum gap shifts dramatically for small changes to $C$, and as can be seen from \cite{roland2002quantum}, devising near-optimal schedules may require locating the minimum gap to within exponentially small error. It is also clear from \cref{fig:problems} that arbitrary optimization problems do not fit the profile of Grover search. A general-purpose strategy that attempts to use the same schedule for all cost functions therefore seems untenable.
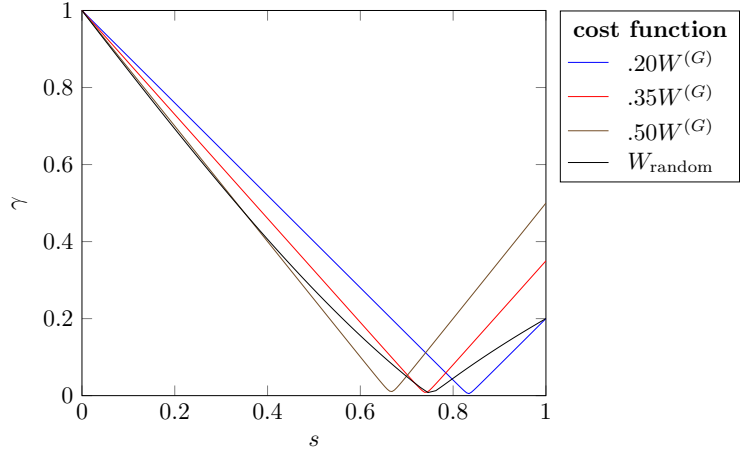
\begin{figure}\centering
\newcommand{\N}{4096}
\begin{tikzpicture}[scale=0.9]
    \begin{axis}[xmin=0,xmax=1,ymin=0,ymax=1,xlabel=$s$,ylabel=$\gamma$,legend pos=outer north east]
    \addlegendimage{empty legend}
    \addlegendentry{\hspace{-.6cm}\textbf{cost function}}
      \foreach \w in {.20,.35,.50}
      {
        \addplot plot[samples=200,domain=0:1, mark=none] function {((1-(1-\w)*x)**2 - 4*\w*x*(1-x)*(\N-1)/(\N))**(.5)}; \expandafter\addlegendentry\expandafter{\w$W^{(G)}$}
      }
        \addlegendentry{$W_{\text{random}}$}
      \addplot plot[no marks, black] coordinates{
      (0.000000,1.000000) (0.016949,0.973431) (0.033898,0.946943) (0.050847,0.920540) (0.067797,0.894226) (0.084746,0.868007) (0.101695,0.841887) (0.118644,0.815872) (0.135593,0.789967) (0.152542,0.764180) (0.169492,0.738516) (0.186441,0.712982) (0.203390,0.687586) (0.220339,0.662334) (0.237288,0.637236) (0.254237,0.612299) (0.271186,0.587532) (0.288136,0.562946) (0.305085,0.538548) (0.322034,0.514351) (0.338983,0.490363) (0.355932,0.466596) (0.372881,0.443061) (0.389831,0.419770) (0.406780,0.396733) (0.423729,0.373964) (0.440678,0.351473) (0.457627,0.329272) (0.474576,0.307374) (0.491525,0.285789) (0.508475,0.264529) (0.525424,0.243604) (0.542373,0.223023) (0.559322,0.202797) (0.576271,0.182932) (0.593220,0.163437) (0.610169,0.144318) (0.627119,0.125580) (0.644068,0.107229) (0.661017,0.089268) (0.677966,0.071707) (0.694915,0.054560) (0.711864,0.037870) (0.728814,0.021808) (0.745763,0.008031) (0.762712,0.012869) (0.779661,0.027303) (0.796610,0.042100) (0.813559,0.056716) (0.830508,0.071066) (0.847458,0.085133) (0.864407,0.098916) (0.881356,0.112422) (0.898305,0.125658) (0.915254,0.138634) (0.932203,0.151362) (0.949153,0.163852) (0.966102,0.176114) (0.983051,0.188160) (1.000000,0.200000)
      };
   \end{axis}
\end{tikzpicture}
\caption{\label{fig:problems} The spectral gap as a function of $s$ for a few Grover-type problems and a random optimization problem. Here, $W_{\text{random}}$ represents the random problem, while $W^{(G)}$ is the Grover cost function. Note that although $W$ is initially close to $0.50W^{(G)}$, it rapidly diverges and is not well approximated by any Grover-type problem. We have used $V = 2^{12}$ vertices (or $12$ qubits).}
\end{figure} 
We overcome these difficulties by designing an efficient algorithm for the gap oracle, and using it in BAA to find an appropriate schedule. 
Unlike in the approach of \cite{roland2002quantum}, we will not be given the promise that $W$ has at most two unique, known eigenvalues. Instead, indexing the vertices so that $W_{u_0}< W_{u_1} \leq \dots \leq W_{u_{V-1}}$, we require the following weaker promises: 
\begin{enumerate}
    \item $W_{u_0} = 0$,
    \item a lower bound on the spectral gap $W_{u_1} - W_{u_0} = W_{u_1}$,
    \item an upper bound on the spectral ratio $W_{u_{V-1}}/W_{u_1}$, and
    \item $\norm{W}\leq V$.
\end{enumerate}
Assumption 2 is necessary to guarantee that the adiabatic process has a finite spectral gap at or near $s=1$. Cases in which Assumption 2 is violated may be interesting, however these are cases in which the adiabatic process itself cannot be guaranteed to be successful by a gap-dependent adiabatic theorem. Assumption 3 merely facilitates our analysis, guaranteeing that the spectral ratio does not increase with $V$. That is, since we work with the combinatorial Laplacian where $\lambda_\max = V$, we require that $W: \mathcal{V} \longrightarrow [0,V]$, or $W$ is merely a rescaled version of ${W}/{\lambda_\max}:V \longrightarrow [0,1]$. Analyzing other situations would be interesting and presumably require only modest, but unnecessarily technical adaptations of the theorems that follow. Restrictions on the distribution of $W$ can also yield better scaling, such as in the case that $W$ is proportional to the Grover cost function; however, we focus on the most general case here. \footnote{For an algorithm that achieves better scaling in these instances, see \cref{alg:optimal}.}

Assumption 1 is the only assumption of which the reader should be suspicious. It can actually be relaxed slightly to $0\leq W_{u_0} < \epsilon V^{2/3}$. On the scale of the problem, this ends up being an exponentially small distinction, so we will proceed with Assumption 1 and leave generalizations to the interested reader. Combined with Assumption 2, this also guarantees that $W$ has a unique minimum. 

In contrast to the first three assumptions, Assumption 4 is quite arbitrary and just simplifies our presentation, since it imposes the constraint that $\norm{W} \leq \norm{L}=V$. An alternative presentation might use the fact that $\norm{W} = \kappa(W)\gamma(W)$ and a similar analysis would follow, however the theorems become a bit more cumbersome. In particular, as long as $\gamma(W)/\norm{W}$ is larger than some constant, the relaxation of Assumption 4 is trivial and left to the reader.

In \cref{sec:the_oracle}, we propose an algorithm (\cref{alg:cg_oracle}) for the gap oracle \textproc{GetGap} in \cref{alg:BAA} and discuss its behavior. The remainder of this section is then dedicated to analyzing the behavior of \cref{alg:BAA} using \cref{alg:cg_oracle} as \textproc{GetGap}.

As we will see in \cref{sec:spectral_theory}, an efficient oracle will require us to derive a tighter Cheeger inequality particular to the complete graph. This and other useful facts arising from the spectral theory of the complete graph will be explored in \cref{sec:spectral_theory}. In \cref{sec:statistics}, we will show how to apply these tools to build our oracle. Finally, in \cref{sec:runtime} we will analyze the full runtime of \cref{alg:BAA} using the oracle of \cref{alg:cg_oracle}.

\subsection{Explicit construction of a gap oracle}\label{sec:the_oracle} 

In this section, we present a function \textproc{GetGap} to be used in \cref{alg:BAA} when $H(0)=L$, the combinatorial Laplacian of the complete graph on $V$ vertices. After some balancing of parameters, each query to our oracle \textproc{GetGap} requires at most time $\bigO{V^{2/3}}$, where the actual time depends upon the ratio of the largest to second smallest eigenvalues of $H(1)$. Thus, we seek only similar scaling from \cref{alg:BAA}. 
\clearpage
\begin{algorithm}
 \captionof{algorithm}{\label{alg:cg_oracle}Complete graph oracle}
 \begin{algorithmic}[1]
    \Require A failure probability $p$, the number of vertices $V$, the cost function $W$, a lower bound $\chi V \leq W_{u_1}$, an upper bound $\kappa \geq W_{u_{V-1}}/W_{u_1}$, a universal constant $c_0 \in (0,1)$
 \item[]
 \State Global $\mathcal{W} \gets \emptyset$ \Comment{Initialize $\mathcal{W}$ to be empty}
 \State Global $S_\min \gets 0$
 \State \label{alg:cg_oracle:xmin} Global $x_\min \gets \max\left\{\dfrac{1}{\kappa^3}[\kappa-1)(V-1)]^{2/3},2(1+c_0)\sqrt{V} \right\}$
 \State \label{alg:cg_oracle:n} Global $n \gets \max\left\{\left\lceil \left(\dfrac{1+c_0}{1-c_0}\right)^2\dfrac{5(V-1)^2(\kappa-1)^2}{8c_0^2  x_\min^2}\ln\left(\dfrac{2}{p}\right)\right\rceil,1\right\}$  \Comment{Choose $n$ by \cref{thm:oracle_construction}}
 \item[]
 \Function{GetGap}{$s,\delta s,\gamma$} \label{alg:cg_oracle:start}
            \If{$S_\min > 0$} \Return \Call{FinishSchedule}{$s,\delta s,\gamma$} \label{alg:cg_oracle:sm_check} \EndIf
            \State $x_0 \gets (1+c_0)\left(\dfrac{\gamma}{1-s} + 1\right)$ \Comment{Extract $x_0$ from previous gap bound} 
            \If{$s = 0$} \label{alg:cg_oracle:x0} $x_0 \gets V$\EndIf
            \State  $x_0 \gets \Call{FindRoot}{s,\delta s,x_0}$
            \If{$x_0 = 0$} \Return \Call{FinishSchedule}{$s,\delta s, \gamma$} \EndIf
            \State \Return $(1-s-\delta s)\left(\dfrac{x_0}{1+c_0}-1\right)$  \Comment{Lower bound the gap using \cref{thm:cheeger,thm:oracle_construction}}
            \label{alg:cg_oracle:start_end}
 \EndFunction
 \item[]
 \Function{$\widetilde\Theta$}{$s,x$}
    \While{$\abs{\mathcal{W}} < n$} \label{alg:cg_oracle:populate} \Comment{Populate $\mathcal{W}$}
        \State $w \gets 0$
        \While{$w = 0$} \label{alg:cg_oracle:w0}  $w \gets \bra{v}W\ket{v}$ for a random vertex $v \in \mathcal{V}$
        \EndWhile
        \State Append $w$ to $\mathcal{W}$
    \EndWhile 
 \State \Return $\dfrac{V-1}{n}\sum_i \left(\dfrac{s}{1-s}\mathcal{W}_i + x\right)^{-1} +\dfrac{1}{x} - 1$ \label{alg:cg_oracle:17}
 \EndFunction
 \item[]
 \Function{FindRoot}{$s,\delta s,x_0$} 
 \If{$\widetilde\Theta(s + \delta s, x_\min)  < 0$ or $x_0 \leq x_\min$}    \label{alg:cg_oracle:if}
 \State \label{alg:cg_oracle:sm}$S_\min \gets s+\dfrac{4(1-s) x_\min}{(1-c_0)^2\chi V}$  \Comment{Upper bound $s_\min$ using \cref{cor:cheeger,cor:smin}}
 \State \label{alg:cg_oracle:sm_end} \Return $0$ 
 \EndIf 
 \State $I \gets \left[\dfrac{(1-c_0)^2}{1+c_0}x_0,\dfrac{(1+c_0)^2}{1-c_0}x_0\right]$ \Comment{Interval from \cref{prop:interval}}
 \State \Return $x \in I$ such that $\widetilde\Theta\left(s+\delta s ,x\right) \approx 0$. \Comment{Root of \cref{eqn:theta_delta}}
 \EndFunction
 \item[]
 \Function{FinishSchedule}{$s,\delta s,\gamma$} 
 \If {$s \leq S_\min$} 
    \State $\gamma \gets \max\left\{(1-c_0)\gamma ,(1-s-\delta s)\sqrt{V-1}/{\kappa^4} \right\}$\Comment{Bound by \cref{thm:weyl,prop:large_side}} \label{alg:cg_oracle:near_minimum}
    \State \Return $\gamma$ \label{alg:cg_oracle:smin}  
 \EndIf
 \State \label{alg:cg_oracle:envelope} $\gamma \gets \dfrac{\chi(V-2)}{4\kappa^5}  (s+\delta s-S_\min) + \dfrac{V-2}{2\kappa^4\sqrt{V-1}}(1-S_\min)$ \Comment{Use linear lower bound from \cref{thm:linear-envelope}}
 \State \Return $\gamma$ \label{alg:cg_oracle_finish_end}
 \EndFunction    
 \end{algorithmic}
\end{algorithm}

\clearpage

The algorithm consists of three major parts. Prior to the gap minimum when the ground state has small amplitude on all vertices, we use a classical root-finding method to determine the appropriate Cheeger constant of $H$. While nearing the minimum $s_\min$, we hit a region where $s < s_\min$ and the root-finding algorithm is no longer an efficient method for determining the Cheeger constant to desired precision. At this point, we use an analytic lower bound on the gap and bound the parameter $s_\min$ determined in the first half of the algorithm until we are guaranteed that $s > s_\min$. Finally, we use a linear envelope to complete the schedule.

\subsection{Spectral graph theory}\label{sec:spectral_theory}

Consider a Hamiltonian $H = L + W$ where $L$ is the combinatorial graph Laplacian of the unweighted complete graph on $V$ vertices and $W$ is any matrix diagonal in the basis of vertices, with diagonal entries $W_u$. Such a Hamiltonian is stoquastic, and has non-negative eigenvalues and real eigenvectors. Let $\lambda_0 < \lambda_1 \leq \dots \leq \lambda_{V-1}$ be the eigenvalues of $H$ with corresponding normalized eigenvectors $\phi_0,\phi_1,\dots,\phi_{V-1}$, and let $\gamma \coloneqq \lambda_1 - \lambda_0$ denote the spectral gap of $H$. (It follows from the Perron-Frobenius theorem that $\lambda_0 <\lambda_1$ and that we can always choose the ground state $\phi_0$ to have strictly positive components.) By definition, the eigenvectors satisfy the equations
\begin{equation}\label{eqn:eigenvectors}
    {(V + W_u -\lambda_i) \phi_i(u) = \sum_{v} \phi_i(v)},
\end{equation} 
for all vertices $u$, where $\phi_i(u)$ denotes the component of $\phi_i$ corresponding to $u$ and the summation is over the set $\mathcal{V}$ of all vertices in the graph.
For a subset $\mathcal{S} \subseteq \mathcal{V}$, we define the ratio
\begin{equation}\label{eqn:not_cheeger}
    g_{\mathcal{S}} \coloneqq \frac{\sum_{{u\in \mathcal{S}, v \notin \mathcal{S}}}\phi_0(u)\phi_0(v)}{\sum_{u \in \mathcal{S}} \phi_0(u)^2}.
\end{equation}
Then, the \textit{Cheeger ratio} corresponding to $\mathcal{S}$ is given by (see \cite{Chung2000})
\begin{equation}\label{eqn:cheeger_ratio}
    h_{\mathcal{S}} = \max\{g_{\mathcal{S}}, g_{\mathcal{V}\setminus \mathcal{S}}\}. 
\end{equation}
and the \textit{(weighted) Cheeger constant} of the Hamiltonian $H$ is
\begin{equation}\label{eqn:cheeger_constant}
    h = \min_{\mathcal{S}\subset \mathcal{V}} h_{\mathcal{S}}.
\end{equation} 
It was shown in \cite{jarret2018hamiltonian} that a general stoquastic Hamiltonian $H=L+W$, where the Laplacian $L$ need not be that of the complete graph, obeys the inequality
\begin{equation}\label{eqn:cheeger_inequality}
    2h \geq \gamma \geq \sqrt{h^2+d^2} - d
\end{equation}
where $d$ is an upper bound on the degree of the graph corresponding to $L$. (For the complete graph, $d=V-1$.) 

For $h \sim V$, this inequality gives us relatively tight control over the spectral gap of $H$, but if $h$ is small, the inequality can be quadratically loose. This problem is potentially unique to graphs of exponentially large degree. If, on the other hand, $d$ is upper bounded by a constant independent of the problem size, then \cref{eqn:cheeger_inequality} implies that $\gamma(H)\sim h$.

In our case, $d=V-1$ and, although potentially possible to obtain, we cannot prove scaling better than $\bigO{({V}/{\gamma_\min})^2} = \bigO{V}$ out of \cref{alg:BAA} using the tools of \cref{sec:adiabatic}. Thus, we wish to derive a Cheeger inequality \textit{specific to the complete graph}. A Cheeger inequality for a particular graph makes no sense classically, where the Cheeger constant maps each graph to a number. In particular, classical Cheeger inequalities need to be flexible enough to apply to \textit{any} graph. In our setting, however, all Cheeger ratios are functions of the cost function $W$ and the graph itself is treated as a known parameter, so we can use information about graph of interest and derive Cheeger inequalities special to that graph. These inequalities, by virtue of the fact that they have been tailored to a particular graph, can be much tighter than those that are expected to work on all graphs.

\subsubsection{Cheeger inequalities for the complete graph} \label{5.1}
The following fact about the second smallest eigenvector of $H$ will allow us to derive a tighter Cheeger inequality in the special case of a complete graph. 
\begin{prop}\label{thm:negative}
    Suppose that $W$ has a unique smallest eigenvalue. Then, there exists a unique $u$ such that $\phi_1(u) \phi_1(v) < 0$ for all $v \neq u$.
\end{prop}
\begin{proof}
    We index the vertices as $u_0, u_1, \dots, u_{V-1}$ such that $W_{u_0}< W_{u_1} \leq \dots \leq W_{u_{V-1}}$. Using the Rayleigh quotient,
    \begin{equation*}
        \lambda_1 = \inf_{f\perp \phi_0}\frac{\langle f, Hf\rangle}{\langle f, f\rangle} = \inf_{f \perp \phi_0} \frac{ \sum_{\{u,v\}}\left(f(u)-f(v)\right)^2 + \sum_u f^2(u)W_u}{\sum_u f^2(u)}.
    \end{equation*}
    Taking $f(u_0) = -\phi_0(u_1)$, $f(u_1)=\phi_0(u_0)$, and $f(u) = 0$ for all $u \neq u_0, u_1$ demonstrates that $\lambda_1 \leq V + W_{u_1}$ with equality only if $W_{u_0}=W_{u_1}$. 
    Since $W_{u_0}< W_{u_1}$ by assumption, the inequality is strict. It then follows from \cref{eqn:eigenvectors} that $\sum_v\phi_1(v) \neq 0$ and that $\phi_1(u_k)\sum_v\phi_1(v) > 0$ for all $k \geq 1$. This in turn implies that $\phi_1(u_0)\sum_v\phi_1(v) < 0$; otherwise, all of the components of $\phi_1$ would have the same sign, contradicting $\phi_1 \perp \phi_0$ (recall that we can choose $\phi_0 > 0$ by the Perron-Frobenius theorem). Therefore, $\phi_1(u_0)\phi_1(v) < 0$ for all $v \neq u_0$.
\end{proof}
Hereafter, we label $m \coloneqq u_0$, since this will correspond to the ``marked" state that our algorithm aims to find.
\begin{thm}\label{thm:cheeger}
    Suppose that $W$ has a unique smallest eigenvalue. Letting $W_{m}< W_{u_1} \leq \dots \leq W_{u_{V-1}}$ denote the eigenvalues of $W$, if $W_{m} = 0$ and $W_{u_{V-1}}/W_{u_1} \leq \kappa$, then 
    \begin{equation*}
        \gamma \geq \max\left\{g_{\{m\}}, \frac{h_{\{m\}}}{\kappa^3} \right\}.
    \end{equation*}
\end{thm}
\begin{proof}
As shown in the proof of \Cref{thm:negative}, $\phi_1(m)\phi_1(v) < 0$ for all $v \neq m$. By \cref{eqn:eigenvectors},
\begin{dgroup*}
\[
    {(V+W_u - \lambda_0)\phi_0(m)\phi_1(m) = \sum_{v}\phi_0(v)\phi_1(m)}
\]
\[
    {(V+W_u - \lambda_1)\phi_1(m)\phi_0(m) = \sum_{v}\phi_1(v)\phi_0(m)},
\]
\end{dgroup*}
whence
\[
    \gamma \phi_0(m)\phi_1(m) = \sum_v \left(\phi_0(v)\phi_1(m)-\phi_0(m)\phi_1(v)\right).
\]
Letting $f(u) \coloneqq \phi_1(u)/\phi_0(u)$ for all $u$, we have
\[
    {\gamma \phi_0(m)^2 = \frac{1}{f(m)}\sum_{v\neq m} \phi_0(v)\phi_0(m)(f(m)-f(v)) = \phi_0(m)\sum_{v\neq m}\phi_0(v)\left( 1-\frac{f(v)}{f(m)}\right)}.
\]
By \cref{eqn:not_cheeger}, $g_{\{m\}} = \sum_{v\neq m}\phi_0(v)/\phi_0(m)$, and, noting that $f(v)/f(m) < 0$ for all $v \neq m$ by \cref{thm:negative}, we have 
\begin{equation*}
    \gamma\phi_0(m)^2 \geq \phi_0(m)\sum_{v\neq m}\phi_0(v) = \phi_0(m)^2 g_{\{m\}},
\end{equation*}
so $\gamma \geq g_{\{m\}}$. 

Since the above lower bound becomes loose for large $\phi_0(m)$, we consider the case where $\phi_0(m) \geq 1/\sqrt{2}$, so that $\phi_0(m)^2 > 1-\phi_0(m)^2$ and hence $h_{\{m\}} = \phi_0(m)\sum_{v\neq m}\phi_0(v)/(1-\phi_0(m)^2)$ by \cref{eqn:cheeger_ratio}. Observing from 
\cref{eqn:eigenvectors} that $\min_{u\neq m}\phi_0(u) = \phi_0(u_{V-1})$ and that $\lambda_0 < V$, we see that
for any $v\neq m$, 
\begin{equation} \label{eqn:kappa_bound}
\frac{\min_{u\neq m}\phi_0(u)}{\phi_0(v)} \geq \frac{\phi_0(u_{V-1})}{\phi_0(v)} = \frac{V+W_v - \lambda_0}{V+ W_{u_{V-1}} - \lambda_0} \geq \frac{V+W_{u_1} - \lambda_0}{V + W_{u_{V-1}} - \lambda_0} \geq \frac{W_{u_1}}{W_{u_{V-1}}} \geq \frac{1}{\kappa},
\end{equation}
which also implies that $\phi_0(u)/\phi_0(v) \geq 1/\kappa$ for any $u,v \neq m$.
Thus,
\begin{dgroup*}
\[
    \gamma\phi_0(m)^2\hiderel{ \;}\geq \phi_0(m)\min_{u\neq m}\phi_0(u)\sum_{v\neq m}\left(1 -\frac{f(v)}{f(m)}\right)
\]\[
    \geq \phi_0(m)\left(\frac{1}{V-1}\sum_{u\neq m}\frac{1}{\kappa}\phi_0(u)\right)\left(V - 1 + \frac{1}{|f(m)|}\sum_{v\neq m}|f(v)|\right) 
\]\[    
    = \frac{1}{\kappa}\left(\phi_0(m)\sum_{u\neq m}\phi_0(u)\right)\left(1 + \frac{1}{|f(m)|}\sum_{v\neq m}\frac{|\phi_1(v)|}{(V-1)\phi_0(v)}\right) 
\]\[    
    = \frac{1}{\kappa}h_{\{m\}}(1-\phi_0(m)^2)\left(1 + \frac{1}{|f(m)|}\sum_{v\neq m}\frac{|\phi_1(v)|\phi_0(v)}{\sum_{u\neq m}\phi_0(v)^2}\right) 
\]\[
    \geq \frac{1}{\kappa}h_{\{m\}}(1-\phi_0(m)^2)\left(1 + \frac{1}{|f(m)|}\frac{\sum_{v\neq m}|\phi_1(v)|\phi_0(v)}{\sum_{u\neq m}(\kappa\phi_0(u))^2}\right) 
\]\[    
    =\frac{1}{\kappa}h_{\{m\}}(1-\phi_0(m)^2)\left(1 + \frac{\phi_0(m)}{|\phi_1(m)|}\frac{|\phi_1(m)|\phi_0(m)}{\kappa^2(1-\phi_0(m)^2)}\right) 
\]\[    
    = \frac{1}{\kappa}h_{\{m\}}\left(1-\phi_0(m)^2 + \frac{\phi_0(m)^2}{\kappa^2}\right),
\]
\end{dgroup*}
so
\[
    \gamma\hiderel{\;} \geq {\frac{1}{\kappa}h_{\{m\}}\left(\frac{1-\phi_0(m)^2}{\phi_0(m)^2} + \frac{1}{\kappa^2}\right) \geq \frac{h_{\{m\}}}{\kappa^3}.}
\]

\end{proof} 
Since the Cheeger ratio $h_{\mathcal{S}}$ of any subset $\mathcal{S} \subseteq \mathcal{V}$ upper bounds the Cheeger constant $h$, the upper bound in \cref{eqn:cheeger_inequality} implies the following.
\begin{cor}\label{cor:cheeger}
Under the conditions of \cref{thm:cheeger},
\begin{equation*}
    2 h_{\{m\}} \geq \gamma \geq \frac{h_{\{m\}}}{\kappa^3}.
\end{equation*}
\end{cor}
\cref{cor:cheeger} demonstrates that the Cheeger ratio $h_{\{m\}}$, corresponding to the cut that isolates the marked state $m$, determines the gap to within a constant factor whenever $\kappa$ is a constant independent of $V$.

For ease of presentation in what follows, we will write $\phi \coloneqq \phi_0$ for the ground state of $H$. Noting that
\begin{equation*}
    g_{\{m\}} = \frac{\norm{\phi}_1-\phi(m)}{\phi(m)} = \frac{\norm{\phi}_1}{\phi(m)} - 1,
\end{equation*}
it will be convenient to introduce the notation
\begin{equation*}
    X \coloneqq \frac{\norm{\phi}_1}{\phi(m)} = g_{\{m\}} + 1,
\end{equation*}
where $m$ is the vertex corresponding to the smallest eigenvalue of $W$.

We now prove a couple of useful facts about $X$.
\begin{prop}\label{prop:simple_facts}
Suppose that $W_m=0$ is the unique smallest eigenvalue of $W$. Then,
\begin{enumerate}
        \item $X = V - \lambda_0$ and
        \item $\sum_u (X + W_u)^{-1} = 1$ \label{prop:sf_fact2}.
    \end{enumerate}
\end{prop}
\begin{proof}
Since $W_m = 0$, \cref{eqn:eigenvectors} gives
\begin{equation*}
(V - \lambda_0)\phi(m) = \norm{\phi}_1,
\end{equation*}
so
\begin{equation*}
    V -\lambda_0 = \frac{\norm{\phi}_1}{\phi(m)} = X,
\end{equation*}
which establishes Fact 1. 

Then,
\begin{equation}
    \norm{\phi}_1 = \sum_u\phi(u) = \sum_u\frac{\norm{\phi}_1}{V-\lambda_0 + W_u} = \norm{\phi}_1 \sum_u \left(X + W_u  \right)^{-1},
\end{equation}
and Fact 2 follows upon dividing both sides by $\norm{\phi}_1$.
\end{proof}

Using the above results, we can obtain a perturbative bound for $g_{\{m\}}$. For $s\in [0,1)$, consider the Hamiltonian
\begin{align*} H(s) &= (1-s)L + sW = (1-s)\left(L + \frac{s}{1-s}W\right) \eqqcolon (1-s)G(s)
\end{align*}
and let $g_{\mathcal{S}}(s)$, $h_{\mathcal{S}}(s)$, and $h(s)$ denote the quantities defined in \cref{eqn:not_cheeger,eqn:cheeger_ratio,eqn:cheeger_constant} corresponding to $G(s)$. It is clear that the results we have proven for $H=L+W$ extend directly to $G(s)$.

\begin{prop}\label{prop:h_change} 
Suppose that $W_m=0$ is the unique smallest eigenvalue of $W$ and that $\|W\| \leq V$. If $0 \leq \delta s \leq \dfrac{c_0}{4V}g_{\{m\}}(s)(1-s)$ for some $c_0 \in (0,1)$, then 
    \begin{equation*}
        {\abs*{g_{\{m\}}(s+\delta s) - g_{\{m\}}(s)} \leq c_0 g_{\{m\}}(s).}
    \end{equation*}
\end{prop}
\begin{proof}
For any matrix $M$, let $\lambda_0(M)$ denote the smallest eigenvalue of $M$. 
By \cref{prop:simple_facts}, 
\begin{align*}
    |g_{\{m\}}(s+\delta s) - g_{\{m\}}(s)| &= |(V-\lambda_0(G(s+\delta s)) - (V-\lambda_0(G(s)))| \\
    &= \left|\lambda_0\left(\frac{H(s+\delta s)}{1-(s+\delta s)}\right) - \lambda_0\left(\frac{H(s)}{1-s}\right)\right| \\
    &=\frac{1}{1-s-\delta s}\left|\lambda_0(H(s+\delta s)) - \left(1-\frac{\delta s}{1-s}\right)\lambda_0(H(s)) \right| \\
    &\leq \frac{1}{1-s-\delta s}\left(|\lambda_0(H(s+\delta s)) - \lambda_0(H(s))| + \delta s \lambda_0(G(s))\right) \\
    &\leq \frac{\delta s}{1-s-\delta s}\left(2 V + (V-g_{\{m\}} - 1)\right) \\
    &\leq \frac{\delta s}{1-s-\delta s}(3V) \\
    &\leq \frac{3Vc_0g_{\{m\}}(s)}{4V - c_0g_{\{m\}}(s)} \\
    &\leq \frac{3Vc_0g_{\{m\}}(s)}{4V - c_0V} \\
    &=\frac{3c_0g_{\{m\}}(s)}{4-c_0} \\
    &\leq c_0g_{\{m\}}(s),
\end{align*}
We arrive at the second inequality by applying Weyl's inequality to $H(s+\delta s) = H(s) + \delta s (-L + W)$, giving
\begin{equation*}
    |\lambda_0(H(s+\delta s)) - \lambda_0(H(s))| \leq \delta s \|-L + W\| \leq \delta s(2V) 
\end{equation*}
since $\|W\| \leq V$ by assumption. The third inequality follows from our assumption that $\delta s/(1-s) \leq c_0g_{\{m\}}(s)/4V$, and the fourth inequality from the fact that $0 \leq g_{\{m\}} = V - \lambda_0 - 1 \leq V$, by \cref{prop:simple_facts}.
\end{proof} 
\subsubsection{Bounds on \texorpdfstring{$h_{\{m\}}$}{hm}}
In order to apply the result of \cref{cor:cheeger} in our algorithm, we require analytic bounds on the Cheeger ratio $h_{\{m\}}$ corresponding to the marked vertex.

\begin{prop}\label{prop:norm_bound}  Under the conditions of \cref{thm:cheeger}, 
    \begin{equation*}
     \sqrt{V-1} \sqrt{1-\phi(m)^2} \geq \sum_{u\neq m}\phi(u) \geq  \frac{1}{\kappa}{\sqrt{V-1}} \sqrt{1-\phi(m)^2}.
    \end{equation*}
\end{prop}

\begin{proof}
    The upper bound follows from Holder's inequality:
    \begin{equation*}
       \sum_{u\neq m}\phi(u) = \frac{\sum\limits_{u\neq m}\phi(u)}{\sqrt{\sum\limits_{u\neq m} \phi(u)^2}}\sqrt{\sum_{u\neq m} \phi(u)^2}
        \leq \sqrt{V-1} \sqrt{1-\phi(m)^2}.
    \end{equation*}
    For the lower bound, we note that $\min_{v\neq m}\phi_0(v)$ is achieved by $u_{V-1}$ and that, by \cref{eqn:kappa_bound}, $\min_{v\neq m}\phi(v)^2 \geq \phi(u)^2/\kappa^2$ for all $u \neq m$. Hence,
    \begin{equation*} \sum_{u\neq m}\phi(u) \geq\sqrt{V-1}\sqrt{\sum_{u\neq m}\min_{v\neq m}\phi(v)^2} \geq\sqrt{V-1}\sqrt{\sum_{u\neq m}\frac{1}{\kappa^2}\phi(u)^2}= \frac{1}{\kappa}\sqrt{V-1}\sqrt{1-\phi(m)^2}. \end{equation*} \end{proof} 
Although $h_{\{m\}} = g_{\{m\}}$ whenever $\phi(m)^2 \leq 1/2$, when $\phi(m)^2 > 1/2$ we would still like to express $h_{\{m\}}$ analytically in terms of the amplitude $\phi(m)$. We exploit the fact that $h_{\{m\}} = g_{\{m\}} \max\left\{1,\frac{\phi(m)^2}{1-\phi(m)^2}\right\}$ to obtain the following analytic bound on $h_{\{m\}}$ for any $\phi(m)$.
\begin{prop}\label{prop:large_side}
    Under the conditions of \cref{thm:cheeger}, if $\phi(m) \in (0,1)$, then
    \begin{equation*}
        \sqrt{V-1}\Phi(m)  \geq h_{\{m\}} \geq \frac{1}{\kappa}\sqrt{V-1}\Phi(m)
    \end{equation*}
    where
    \begin{equation} \label{eqn:Phi}
        \Phi(m) \coloneqq \max\left\{\frac{\phi(m)}{\sqrt{1-\phi(m)^2}},\frac{\sqrt{1-\phi(m)^2}}{\phi(m)} \right\}\geq 1. \end{equation}
\end{prop}
\begin{proof}
    By definition,
    \begin{equation*}
        h_{\{m\}} = \frac{\sum\limits_{u\neq m}\phi(u)}{\phi_m}\max\left\{1,\frac{\phi(m)^2}{1-\phi(m)^2}\right\} = \sum_{u\neq m}\phi(u)\frac{\Phi(m)}{\sqrt{1-\phi(m)^2}}.
    \end{equation*}
    The result follows immediately from \cref{prop:norm_bound}.
\end{proof}
  
\subsection{Analysis of BAA on the complete graph}
\label{sec:statistics}

We analyze the runtime of BAA with the oracle constructed in \cref{alg:cg_oracle} as follows. First, we assume that we can query the appropriate Cheeger ratios and show that they obey \cref{thm:queries}. Then, we determine the additional runtime incurred by abandoning queries to the Cheeger ratio when such queries become inefficient.

As in \cref{sec:spectral_theory}, we consider $H(s) = (1-s)L + sW \eqqcolon (1-s)G(s)$ and the corresponding quantities $g_{\mathcal{S}}(s)$, $h_{\mathcal{S}}(s)$, and $h(s)$. While the results of that section apply directly to $G(s) = L + sW/(1-s)$, it is useful to note from \Cref{eqn:not_cheeger,eqn:cheeger_ratio,eqn:cheeger_constant} that $g_{\mathcal{S}}(s)$, $h_{\mathcal{S}}(s)$, and $h(s)$ are functions only of the ground state and are therefore the same for both $H(s)$ and $G(s)$. We also define $s_\min$ to be the point at which $\phi(m) = 1/\sqrt{2}$, so that $g_{\{m\}}(s_\min) = g_{\mathcal{V}\setminus \{m\}}(s_\min)$. Since $\phi(m)=1/\sqrt{V}$ at $s=0$, $\phi(m) = 1$ at $s=1$, and \cref{prop:monotone} shows that $\phi(m)$ is strictly increasing over $s\in[0,1]$, this point is unique.

The following proposition provides a small interval $I$ such that $s_\min \in I$. This interval is useful mostly as a tool for deriving further inequalities. The interval used in \cref{alg:cg_oracle} is smaller than that of \cref{prop:s_min} and will be provided by \cref{cor:smin} in \cref{sec:convex_envelope}.
\begin{prop}\label{prop:s_min}
Suppose that $W_m = 0$ is the unique smallest eigenvalue of $W$, $W_{u_1}$ is the second smallest eigenvalue of $W$, and $\norm{W}\leq V$. Then,
\begin{equation*} s_\min \in\left[\frac{1}{2}\left(1-\frac{1}{V-1}\right),1-\frac{W_{u_1}}{5V}\right].
\end{equation*}
\end{prop}
\begin{proof}
    We begin with the lower bound. For any matrix $M$, let $\lambda_0(M)$ denote the smallest eigenvalue of $M$. 
    Let $W^{(G)}$ be the diagonal matrix with $W_m^{(G)} = 0$ and $W^{(G)}_u = V$ for all $u \neq m$. Using \cref{prop:simple_facts} and noting that $W-W^{(G)}$ is negative semidefinite,
    \begin{align*}
    g_{\{m\}}(s) &= V - \lambda_0(G(s)) - 1 \\
    &= V - \lambda_0\left(L + \frac{s}{1-s}W\right) - 1 \\
    &\geq V - \lambda_0\left(L + \frac{s}{1-s}W^{(G)}\right) - 1 \\
    &= V - \frac{1}{1-s}\lambda_0\left((1-s)L + sW^{(G)}\right) - 1 \\
    &= V - \frac{1}{1-s}\frac{V}{2}\left[1-\sqrt{1-4\frac{V-1}{V}s(1-s)}\right] -1 \\
    &= V\left\{1- \frac{1}{2(1-s)}\left[1-\sqrt{1-4\frac{V-1}{V}s(1-s)}\right] \right\} -1 \\
    &> \sqrt{V-1}
    \end{align*}
    when $s < \frac{1}{2}\left(1-\frac{1}{V-1}\right)$. 
    On the other hand, at $s = s_\min$, $\Phi(m) = 1$ and hence $g_{\{m\}} = h_{\{m\}} \leq \sqrt{V-1}$ by \cref{prop:large_side}. 
    
    For the upper bound, let $\gamma(s)$ be the spectral gap of $H(s)$. Note that $\gamma(1) = W_{u_1}$, and by \cref{thm:weyl} and our assumption that $\|W\| \leq V$,
    \[
        \abs{W_{u_1} - \gamma(s)} = \abs{\gamma(1) - \gamma(s)} \leq 4 (1-s) V,
    \]
    so if $(1-s) < \frac{1}{4V}[W_{u_1} - 2(1-s)\sqrt{V-1}]$, we would have
    \begin{equation*}
        \abs{W_{u_1} - \gamma(s)} < W_{u_1} - 2(1-s)\sqrt{V-1},
    \end{equation*}
    whence $\gamma(s) > 2(1-s)\sqrt{V-1}$. Then, applying \cref{cor:cheeger} to the spectral gap of $G(s)$ implies that $h_{\{m\}}> \sqrt{V-1}$. On the other hand, $h_{\{m\}} \leq \sqrt{V-1}$ at $s = s_\min$ by \cref{prop:large_side}, so we must have $(1-s_\min) \geq \frac{1}{4V}[W_{u_1} - 2(1-s_\min)\sqrt{V-1}]$, or
    \begin{equation*}
        s_\min \leq 1 - \frac{W_{u_1}}{4V + 2\sqrt{V-1}} \leq  1-\frac{W_{u_1}}{5V}. 
    \end{equation*} \end{proof} 
\subsubsection{The convex envelope}\label{sec:convex_envelope}

\begin{thm}\label{thm:linear-envelope}
Let $\gamma(s)$ denote the spectral gap of $H(s) = (1-s)L + sW$. Suppose that $W$ has a unique smallest eigenvalue and that $\|W\| \leq V$. Letting $W_m < W_{u_1} \leq \dots W_{u_{V-1}}$ denote the eigenvalues of $W$, if $W_m = 0$, $W_{u_{V-1}}/W_{u_1} \leq \kappa$, and $W_{u_1} \geq \chi V$ for some constant $\chi \geq 2\sqrt{V-1}/V$, then for all $s \in [0,)$,
\begin{equation*}
    2\kappa^4\oracle(s)\geq \gamma(s) \geq \oracle(s),
\end{equation*}
with
\begin{equation} \label{eqn: oracle}
  \oracle(s)= \begin{dcases}       (1-s)h_{\{m\}} & \text{$s < s_\min$} \\
  (1-s)\frac{\sqrt{V-1}}{\kappa^4} \left(\frac{\phi(m)}{\sqrt{1-\phi(m)^2}} \right) & \text{$s \geq s_\min$}.
  \end{dcases}
\end{equation}
Moreover,
\[
    \oracle(s) \geq  \frac{\chi}{4 \kappa^5}(V-2)\abs*{s-s_\min} + \frac{1}{2}\left(\frac{V-2}{V-1}\right)\oracle(s_\min).
\]
\end{thm}
\begin{proof}
Since $\phi(m)$ is strictly increasing (as show by \cref{prop:monotone} in the Appendix), $\phi(m) < 1/\sqrt{2}$ for $s < s_\min$ and $\phi(m) \geq 1/\sqrt{2}$ for $s \geq s_\min$. Hence, by \cref{thm:cheeger} and \cref{eqn:cheeger_inequality}, the spectral gap $\gamma(G(s))$ of $G(s) = H(s)/(1-s)$ is bounded as $2 h_{\{m\}}(s) \geq \gamma(G(s)) \geq h_{\{m\}}(s)$ for $s < s_\min$ and as $2h_{\{m\}}(s) \geq \gamma(G(s)) \geq h_{\{m\}}(s)/\kappa^3$ for $s \geq s_\min$. Since $h_{\{m\}}(s)$ is invariant under rescaling $H(s)$ by an overall factor while $\gamma(s) = (1-s)\gamma(G(s))$, it follows that when $s < s_\min$,
    $2(1-s)h_{\{m\}}(s) \geq \gamma(s) \geq (1-s)h_{\{m\}}(s)$
or, in terms of $\oracle(s)$,
\begin{equation*}
    2\oracle(s) \geq \gamma(s) \geq \oracle(s).
\end{equation*}
Similarly, when $s > s_\min$, $2(1-s)h_{\{m\}}(s) \geq \gamma(s) \geq (1-s){h_{\{m\}}(s)}/{\kappa^3}$, and applying the bounds on $h_{\{m\}}$ given by \cref{prop:large_side}, we have
\begin{equation*}
    2\kappa^4\oracle(s) \geq \gamma(s) \geq \oracle(s).
\end{equation*}
Thus, for all $s \in [0,1)$, $2\kappa^4\oracle(s) \geq \gamma(s) \geq \oracle(s)$. 

To derive the lower bound on $\oracle(s)$, we consider the two regions separately.

\begin{case}[$s < s_\min$]

\hspace{1.5em} In this region, $h_{\{m\}}(s) = g_{\{m\}}(s)$ and $\oracle(s) = (1-s)g_{\{m\}}(s)$.
Using \cref{prop:simple_facts} and the Hellmann-Feynmann theorem, 
\begin{align}
    \frac{dg_{\{m\}}(s)}{ds} &= \frac{d}{ds}(V - \lambda_0(G(s)) - 1) \nonumber \\
    &=-\bra{\phi}\frac{d}{ds}\left(L + \frac{s}{1-s}W\right)\ket{\phi} \nonumber\\
    &=-\frac{1}{(1-s)^2}\bra{\phi}W\ket{\phi} \label{eqn:g_m derivative}\\
    &=-\frac{1}{(1-s)^2}\sum_{u}\phi(u)^2W_u \nonumber\\
    &\leq - \frac{1}{(1-s)^2}(1-\phi(m)^2)W_{u_1} \nonumber\\
    &\leq -\frac{W_{u_1}}{2(1-s)^2},\nonumber
\end{align}
where the first inequality follows from the assumption that $W_m=0$ and $W_u \geq W_{u_1}$ for all $u\neq m$, and the second from the fact that $\phi(m)^2 <1/2$ in this region. Integrating both sides over $[s,s_\min]$ for some $s<s_\min$, we find
\begin{equation*}
    g_{\{m\}}(s) - g_{\{m\}}(s_\min) \geq \frac{W_{u_1}}{2}\left(\frac{1}{1-s_\min} - \frac{1}{1-s}\right).
\end{equation*}
Thus
\begin{dgroup*}
\[ \Gamma(s) = (1-s)g_{\{m\}}(s) \]
\[\geq (1-s_\min)g_{\{m\}}(s) \]
\[ \geq \frac{W_{u_1}}{2}\frac{s_\min - s}{1-s} + (1-s_\min)g_{\{m\}}(s_\min)\]
\[ \geq \frac{\chi V}{2}|s-s_\min| + \lim_{s\to s_\min^{-}}\Gamma(s)
\]
\end{dgroup*}
where the last line follows from the assumption that $W_{u_1} \geq \chi V$.
\end{case} 

\begin{case}[$s \geq s_{\min}$]
\hspace{1.5em} In this region, $\oracle$ is given by
\[
    \oracle(s) = (1-s)\frac{\sqrt{V-1}}{\kappa^4}\frac{\phi(m)}{\sqrt{1-\phi(m)^2}} = (1-s)\frac{\sqrt{V-1}}{\kappa^4}\frac{1}{\sqrt{X(s)^2\sum\limits_{u\neq m}\left(X(s) + \frac{s}{1-s}W_u\right)^{-2}}} \geq \frac{\sqrt{V-1}}{\kappa^4}\frac{s}{X(s)\sqrt{\sum\limits_{u\neq m}W_u^{-2}}},
\]
where we used the fact that
\begin{equation*}
    \frac{1-\phi(m)^2}{\phi(m)^2} = \frac{1}{\phi(m)^2}\sum_{u\neq m}\left(\frac{\norm{\phi}_1}{V-\lambda_0(G(s)) + \frac{s}{1-s}W_u}\right)^2 = X(s)^2\sum_{u\neq m}\left(X(s) + \frac{s}{1-s}W_u\right)^{-2}
\end{equation*}
by \cref{eqn:eigenvectors} and \cref{prop:simple_facts}, writing $X(x) \coloneqq g_{\{m\}}(s) + 1$. Hence, we define
\begin{equation*}
    \widetilde{\oracle}(s) \coloneqq \frac{\sqrt{V-1}}{\kappa^4}\frac{s}{X(x)\sqrt{\sum\limits_{u\neq m}W_u^{-2}}}
\end{equation*}
as the lower bound on $\oracle(s)$, and consider the derivative
\[
    \frac{d}{ds}\left(\frac{s}{X(s)}\right) = \frac{1}{X(s)} - \frac{s}{X(s)^2}\frac{dX(s)}{ds}
    \geq -\frac{s}{X(s)^2}\frac{dX(s)}{ds}
    = \frac{s}{X(s)^2}\frac{1}{(1-s)^2}\bra{\phi}W\ket{\phi}
    \geq \frac{s}{X(s)^2}\frac{1}{(1-s)^2}(1-\phi(m)^2)W_{u_1}
    =\frac{sW_{u_1}}{X(s)^2(1-s)^2}\left[\phi(m)^2X(s)^2\sum_{u\neq m}\left(X(s) +\frac{s}{1-s}W_u\right)^{-2}\right]
    \geq \frac{sW_{u_1}}{(1-s)^2}\frac{1}{2}\sum_{u\neq m}\left(W_u + \frac{s}{1-s}W_u\right)^{-2}
    =\frac{sW_{u_1}}{2(1-s)^2}\sum_{u\neq m}\left(\frac{1}{1-s}W_u\right)^{-2}
    \geq \frac{sW_{u_1}}{2}\sum_{u\neq m}\left(\frac{1}{W_{u_{V-1}}}\right)^2
    \geq \frac{sW_{u_1}}{2}\sum_{u\neq m}\frac{1}{\kappa W_{u_1}}\frac{1}{V}
    = \frac{s(V-1)}{2\kappa V}.
\]
To obtain the third inequality, we used the fact that $g_{\{m\}}(s)$ is monotonically decreasing in s over $s \in [0,1)$, as is clear from \cref{eqn:g_m derivative} and the assumption that $W$ is positive semidefinite, and that $g_{\{m\}}(s_\min) = h_{\{m\}}(s_\min) \leq \sqrt{V-1}$ by \cref{prop:large_side}. Consequently, for all $u \neq m$
\begin{equation} \label{eqn:X W_u}
    X(s) = g_{\{m\}}(s) + 1 \leq g_{\{m\}}(s_\min) + 1 \leq \sqrt{V-1} + 1 \leq 2\sqrt{V-1} \leq \chi V \leq W_{u_1} \leq W_u.
\end{equation}
The fifth inequality follows from the assumptions that $W_{u_{V-1}} \leq \kappa W_{u_1}$ and $W_{u_{V-1}} \leq V$. Using this to bound $d\Gamma(s)/ds$ and integrating both sides of the resultant expression over $[s_\min, s]$ for some $s \geq s_\min$ gives 
\begin{align*}
\widetilde{\oracle}(s) - \widetilde\oracle(s_\min) &\geq \frac{(V-1)^{3/2}}{2\kappa^5 V\sqrt{\sum\limits_{u\neq m}W_u^{-2}}}\frac{1}{2}(s^2 - s_\min^2) \\
&\geq \frac{(V-1)^{3/2}}{4\kappa^5 V}\sqrt{\frac{W_{u_1}^2}{V-1}}(s + s_\min)|s-s_\min| \\
&\geq \frac{V-1}{4\kappa^5 V}W_{u_1}(2s_\min)|s-s_\min| \\
&\geq  \frac{V-1}{4\kappa^5 V}(\chi V)\left(1-\frac{1}{V-1}\right)|s-s_\min|\\ 
&= \frac{(V-2)\chi}{4\kappa^5}|s-s_\min|
\end{align*}
where the second line follows from the fact that $
    \sum_{u\neq m}W_{u}^{-2} \leq \sum_{u\neq m}W_{u_1}^{-2} = (V-1)/W_{u_1}^2$
and the fourth line follows from the assumption that $W_{u_1} \geq \chi V$ as well as \cref{prop:s_min}. Then, since
\begin{equation*}
    \sum_{u\neq m}W_{u}^{-2} = \frac{1}{(1-s)^2}\sum_{u\neq m}\left(W_u + \frac{s}{1-s}W_u\right)^{-2} \leq \frac{1}{(1-s)^2}\sum_{u\neq m}\left(X(s) + \frac{s}{1-s}W_u\right)^{-2},
\end{equation*}
by \cref{eqn:X W_u},
we have $\widetilde\oracle(s) \geq s\oracle(s)$. Therefore,
\[
    {\oracle(s) \geq 
    \widetilde\oracle(s)} \geq \frac{(V-2)\chi}{4\kappa^5}|s-s_\min| + \widetilde{\oracle}(s_\min)
    \geq \frac{(V-2)\chi}{4\kappa^5}|s-s_\min| + s_\min\oracle(s_\min)
    \geq \frac{(V-2)\chi}{4\kappa^5}|s-s_\min| +\frac{1}{2}\left(\frac{V-2}{V-1}\right)\oracle(s_\min),
\]
using  \cref{prop:s_min} in the last inequality. \\
\end{case}
Since in Case 1,
$\oracle(s) \geq \dfrac{\chi V}{2}|s -s_\min| + \lim_{s\rightarrow s_\min^-}\oracle(s)$,
and $\lim_{s\to s_\min^{-}}\Gamma(s) \geq \Gamma(s_\min)$ by \cref{prop:large_side},
we have that
\[
    \oracle(s) \geq \frac{\chi(V-2)}{4 \kappa^5}\abs*{s-s_\min} + \frac{1}{2}\left(\frac{V-2}{V-1}\right)\oracle(s_\min).
\]
in either case.
\end{proof}

Step \ref{alg:cg_oracle:sm} of \cref{alg:cg_oracle}  requires that we approximate $s_\min$. A sufficient bound follows from the proof of \cref{thm:linear-envelope}.

\begin{cor}\label{cor:smin}
Under the conditions of \cref{thm:linear-envelope}, if $\gamma(s) < \widetilde{\gamma}$ for some constant $\widetilde{\gamma}$ and $s < s_\min$, then 
    \[
        s_\min \leq s+ \frac{2\widetilde\gamma}{\chi V}.
    \]
\end{cor}
\begin{proof}
    This is an immediate consequence of Case 1 in the proof of \cref{thm:linear-envelope}.
\end{proof}

\subsubsection{The \texorpdfstring{$\Theta$}{ϴ} function}
In this subsection, we consider $H = L + W$, where $W$ has eigenvalues $0 = W_m < W_{u_1} \leq \dots \leq W_{u_{V-1}} \leq V$, and determine $X\coloneqq g_{\{m\}}+1$ up to some relative error. Recall that determining $X$ up to relative error will be sufficient to provide a bound on the spectral gap of $H$. According to \cref{prop:simple_facts}, for any such $W$, $X$ is the zero of the function
\begin{equation*}
    \Theta(x) \coloneqq \sum_u(W_u + x)^{-1} - 1.
\end{equation*}
defined on $x \in \mathbb{R}^+$.

If we have access to $\Theta$, then the monotonicity of the function in the variable $x$ implies that the bisection method can rapidly find $X$ to arbitrary error. Although $\Theta(x)$ is simple enough to write down, determining $\Theta(x)$ fully would require knowledge of $W_u$ for every vertex $u$. Hence, \cref{alg:cg_oracle} approximates $X$ by finding the zero $\widetilde{X}$ of the function
\begin{equation}\label{eqn:theta_delta}
    \widetilde\Theta(x) \coloneqq \frac{V-1}{n}\sum_{i=0}^{n-1}\left(W_{y_i} + x\right)^{-1} + \frac{1}{x} - 1,
\end{equation}
where $y_0,y_1,\dots,y_{n-1}$ are i.i.d. random variables with $y_i \sim \mathrm{Uniform}(\mathcal{V}\setminus \{m\})$.
\cref{alg:cg_oracle} takes $\widetilde\Theta(s,x)$ as equivalent to $\widetilde{\Theta}(x)$ with $W\mapsto \frac{s}{1-s}W$, which merely restricts the above expressions to a particular one-parameter family. Since we only seek to understand $X$ as a function of $W$, we suppress the $s$-dependence in this section.

First, we determine how close $\widetilde{\Theta}$ must be to $\Theta$ for $\widetilde{X}$ to be a good estimate of $X$. 
\begin{prop}\label{prop:root} 
For any $x>0$,
    \[
       \abs*{x - X} \leq (W_{u_{V-1}} + x)\abs*{\Theta(x)}.
    \]
\end{prop}
\begin{proof}
Noting that $\Theta(X)=0$ implies $\sum_{u}(W_u +X)^{-1} = 1$,
we have
\begin{align*}
    \left(W_{u_{V-1}} + x \right) \abs*{\Theta(x)} &= \left(W_{u_{V-1}} +x\right)\abs*{\sum_{u}\frac{1}{W_u + x} - 1} \\
    &=(W_{u_{V-1}} + x)\abs*{\sum_{u}\left(\frac{1}{W_u + x} - \frac{1}{W_u + X}\right)} \\
    &= \abs*{x-X}\sum_{u}\frac{W_{u_{V-1}} + x}{(W_u + x)(W_u+X)} \\
    &\geq \abs*{x-X}\sum_u \frac{1}{W_u+X} \\
    &= \abs*{x-X}.
\end{align*}
\end{proof}
\noindent In particular, if $\widetilde{\Theta}$ is such that $|\Theta(\widetilde{X})| \leq \epsilon X/(W_{u_{V-1}} + \widetilde{X})$ for some $\epsilon >0$, then $|X-\widetilde{X}| \leq \epsilon X$. 

Next, we bound $\abs{\delta(x)}$ in terms of the number $n$ of samples drawn from $\mathcal{V}\setminus\{m\}$. The following proposition implies that when $x \approx \widetilde{X}$, we are within the bounds required by \cref{prop:root}.

\begin{prop}\label{prop:hoeffding}
If $\kappa \coloneqq W_{u_{V-1}}/W_{u_1} > 1$ and $n = \left\lceil\dfrac{(V-1)^2(\kappa-1)^2}{2X^2 \epsilon_0^2 }\ln\left(\dfrac{2}{p}\right)\right\rceil$ for some $\epsilon_0 > 0$ and $p\in (0,1]$, then with probability at least $1-p$,
\begin{equation*}
    \abs*{{\Theta}(x) - \widetilde\Theta(x)} \leq \frac{\epsilon_0 X}{W_{u_{V-1}} + x}
\end{equation*}
for all $x>0$.
\end{prop}
\begin{proof}
    First, we note that
    \[
        {\sum_{u} \frac{1}{W_u + x} = \frac{1}{x} + \sum_{u\neq m} \frac{1}{W_u + x} = \frac{1}{x} + \E*{\frac{V-1}{W_{y_i}+x}}.}
    \]
    for any $i$. Hence, using Hoeffding's inequality,
    \[
        \Pr{\abs*{{\Theta}(x) - \widetilde\Theta(x)} \geq t} = \Pr{{\abs*{\sum_u\frac{1}{W_u+x} - \frac{1}{n}\sum_{i=0}^{n-1} \frac{V-1}{W_{y_i} + x} -\frac{1}{x}} \geq t}} 
        = \Pr{{\abs*{\E*{\frac{1}{n}\sum_{i=0}^{n-1}\frac{V-1}{W_{y_i}+x}} - \frac{1}{n}\sum_{i=0}^{n-1} \frac{V-1}{W_{y_i} + x}} \geq t}}
        \leq 2 \exp\left[-2 n t^2\left(\max_{u\neq m}\frac{V-1}{W_u + x}-\min_{u\neq m}\frac{V-1}{W_u + x}\right)^{-2} \right]
        = 2 \exp\left[-\frac{2 n t^2}{(V-1)^2}\left[\frac{(W_{u_{V-1}} +x)(W_{u_1} + x)}{W_{u_{V-1}} - W_{u_1}}\right]^2 \right]
    \]
   for $t \geq 0$. Taking $t = \epsilon_0 X/(W_{u_{V-1}} + x)$, we have
    \[
        \Pr{\abs*{\Theta(x) - \widetilde\Theta(x)} \geq \frac{ \epsilon_0 X}{W_{u_{V-1}} + x}  } \leq 2 \exp\left[-\frac{2n\epsilon_0^2X^2}{(V-1)^2}\left(\frac{W_{u_1} + x}{W_{u_{V-1}} - W_{u_1}}\right)^2\right]
        = 2 \exp\left[-\frac{2n\epsilon_0^2X^2}{(V-1)^2(\kappa-1)^2}\left(1+\frac{x}{W_{u_1}}\right)^2\right]
        \leq 2 \exp\left[-\frac{2 n \epsilon_0^2 X^2}{ (V-1)^2 (\kappa-1)^2} \right]. 
    \]
    Thus, taking 
    $n=\lceil{(V-1)^2 (\kappa-1)^2} \ln(2/p)/({2 \epsilon_0^2 X^2 })\rceil$
    yields the desired result. 
\end{proof}
    The observant reader may worry that while \cref{prop:hoeffding} yields the bound required by \cref{prop:root} when $x = \widetilde{X}$, the bound may be insufficiently tight when $x$ is far from $\widetilde{X}$. Nonetheless, the function $\widetilde{X}$ is monotone decreasing and $\lim_{x\rightarrow 0^+} \widetilde\Theta(x) > 0$ and $\lim_{x \rightarrow \infty} \widetilde\Theta(x) < 0$. Thus, the bisection method, which we discuss in the following subsection, can be used to determine its zero to arbitrary accuracy.
 
\subsubsection{The \textproc{FindRoot} function}
In this subsection, we use $X(s)$ to denote the zero of the function
\begin{equation*}
    \Theta(s,x) \coloneqq \sum_u\left(\frac{s}{1-s}W_u + x\right)^{-1} - 1,
\end{equation*}
defined on $x \in \mathbb{R}^+$ for a given $s\in [0,1)$,
so that $X(s)\coloneqq g_{\{m\}}(s) + 1$ corresponds to  $H(s) = (1-s)L + sW$. Similarly, we write $\widetilde{X}(s)$ for the zero of $\widetilde\Theta(s,x)$, which is defined as in \cref{eqn:theta_delta} but with $W \mapsto \frac{s}{1-s}W$. We note that if we take the number of samples prescribed by \cref{prop:hoeffding}, then $\abs*{\Theta(s,x) - \widetilde{\Theta}(s,x)} \leq \frac{\epsilon_0 X}{\frac{s}{1-s}W_{u_{V-1}} + x}$ with probability $1-p$ for any $s$. Since the vertices are sampled once at the start of the algorithm, $\widetilde{\Theta}$ is constructed using the same $\{W_{y_i}\}$ at every step; consequently, with probability $1-p$, $\widetilde\Theta$ is a good approximation of $\Theta$ for all $s$ such that $X(s) \geq x_\min$, or whenever $\Theta$ gets called. Hence, in this section we assume that $\widetilde\Theta$ is a good approximation of $\Theta$ and do not reference the probability of success.

The \textproc{FindRoot} function in \cref{alg:cg_oracle} approximates $X(s+\delta s)$ as $x_0(s+\delta s)$, using $\widetilde\Theta(s,x)$ and an estimate $x_0(s)$ of $X(s)$ from the previous step. \cref{prop:bisection} determines the number of iterations of the bisection method  required to determine $\widetilde{X}(s+\delta s)$ to within some relative error, when we know that $\widetilde{X}(s+\delta s)$ lies within a certain interval. \cref{prop:interval} demonstrates how to constrain the interval $I$ using $X(s)$. Finally, \cref{thm:oracle_construction} integrates all of these results and guarantees $\abs*{X(s) - x_0(s)} \leq c_0 X(s)$ whenever $X(s) \geq (1-c_0)x_\min$.

$\widetilde{\Theta}(s,x)$ is defined such that it is monotone decreasing over $x \in \mathbb{R}^+$ and, thus, has at most one positive root for any fixed $s$. Furthermore, $\lim_{x\rightarrow 0+}\widetilde{\Theta}(s,x) > 0$ and $\widetilde{\Theta}(s,V+\epsilon) < 0$ for any $\epsilon > 0$. Thus, the bisection method is a natural way to determine this root.

\begin{prop}\label{prop:bisection}
    Suppose that for a given $s$, $\widetilde\Theta$ has a unique zero $\widetilde{X}(s+\delta s)$ in the interval $[a x_0(s), b x_0(s)]$ for some $0<a<b$. Then, the bisection method returns an $x_0(s+\delta s)$ such that $\abs{x_0(s+\delta s) -\widetilde{X}(s+\delta s)} \leq \epsilon_1 \widetilde{X}(s+\delta s)$  using $\ceil*{\log_2[({b}/{a}-1 )/\epsilon_1]}$ evaluations of $\widetilde\Theta$.
\end{prop}

\begin{proof}
Since $\widetilde\Theta$ is monotone and has precisely one zero in the interval $[a x_0(s), b x_0(s)]$, we can apply the bisection method, which returns an estimate $x_0(s+\delta s)$ such that $\abs{x_0(s+\delta s)-\widetilde{X}(s+\delta s)} \leq x_0(s)(b-a)/2^{k}$ after $k$ steps, each of which evaluates $\widetilde\Theta$ once. Thus, if
\[2^k \geq \frac{x_0(s) (b-a)}{\epsilon_1 \widetilde{X}(s+\delta s)}, \]
the method finds an $x_0(s+\delta s)$ for which $\abs{x_0(s+\delta s) - \widetilde{X}(s+\delta s)}\leq \epsilon_1 \widetilde{X}(s+\delta s)$. Since $\widetilde{X}(s+\delta s) \geq ax_0(s)$, it follows that $k=\ceil*{\log_2[({b}/{a}-1 )/\epsilon_1]}$ calls to $\widetilde\Theta$ suffice.
\end{proof} 

\cref{prop:bisection} requires a particular interval such that $\widetilde{X}(s+\delta s) \in [a{x_0(s)},b{x_0(s)}]$ for all $s$ such that $x_0(s) \geq x_\min$. We provide explicit bounds on $a$ and $b$ in the following proposition.

\begin{prop}\label{prop:interval}
    If $\abs{x_0(s) - X(s)}\leq c_0 X(s)$, $0\leq \delta s \leq \dfrac{c_0}{4V}\left(\dfrac{x_0(s)}{1+c_0} -1\right)(1-s)$, and $\abs{\widetilde{X}(s+\delta s) - X(s+\delta s)} \leq \epsilon_0 X(s+\delta s)$, for some $c_0, \epsilon_0, \in (0,1)$, then \[ \widetilde{X}(s+\delta s) \in \left[(1-\epsilon_0)\frac{1-c_0}{1+c_0}x_0(s),(1+\epsilon_0)\frac{1+c_0}{1-c_0}x_0(s)\right]. \] 
\end{prop} 
\begin{proof}
Since $\abs{x_0(s) - X(s)} \leq c_0 X(s)$ implies that $X(s) \geq x_0(s)/(1+c_0)$, the assumption on $\delta s$ ensures that $\delta s \leq {c_0}(X(s) - 1)(1-s)/4V = c_0g_{\{m\}}(s)(1-s)/4V$. Hence, $\abs{X(s+\delta s) - X(s)} \leq c_0X(s)$ by \cref{prop:h_change}. It follows that
\[
    \widetilde{X}(s+\delta s) \geq  (1-\epsilon_0) X(s+\delta s)
    \geq (1-\epsilon_0)(1-c_0) X(s) 
    \geq (1-\epsilon_0)(1-c_0)\frac{x_0(s)}{1+c_0}
\]
and similarly,
\[
    \widetilde{X}(s+\delta s) \leq (1+\epsilon_0) X(s+\delta s) \leq (1+\epsilon_0)(1+c_0) X(s) 
    \leq
    (1+\epsilon_0)(1+c_0)\frac{x_0(s)}{1-c_0}.
\]
\end{proof} 
\begin{thm}\label{thm:oracle_construction}
Suppose that $\kappa \coloneqq W_{u_{V-1}}/W_{u_1} >1$. If for a given $s\in [0,1)$, 
\begin{enumerate}
    \item $|x_0(s) - X(s)| \leq c_0X(s)$ for some $c_0\in (0,1)$,
    \item $|\widetilde{X}(s) - X(s)| \leq \dfrac{9c_0}{10} X(s)$,
    \item $\widetilde{X}(s),x_0(s)> x_\min$ for some constant $x_\min$,
    \item $0\leq \delta s \leq \dfrac{c_0}{4V}\left(\dfrac{x_0(s)}{1+c_0} -1\right)(1-s)$, and
    \item for all $x$,
    \begin{equation} \label{eqn:condition} \abs*{\Theta(s+\delta s,x) - \widetilde{\Theta}(s+\delta s,x)} \leq \frac{9c_0}{10}\left(\frac{1-c_0}{1+c_0}\right)\frac{x_\min}{\dfrac{s}{1-s}W_{u_{V-1}} + x},
\end{equation}
\end{enumerate}
then \textproc{FindRoot}\text{($s$,$\delta s$, $x_0(s)$)} returns an $x_0(s+\delta s)$ such that $|x_0(s+\delta s) - X(s+\delta s)| \leq c_0 X(s+\delta s)$
using at most
\[
    \ceil*{\left(\dfrac{1+c_0}{1-c_0}\right)^2\frac{5(V-1)^2(\kappa-1)^2}{8x_\min^2 c_0^2}\ln\left(\frac{2}{p}\right)}\ceil*{\log_2\left[\frac{19}{c_0}\left(\left(\frac{1+c_0}{1-c_0}\right)^3-1\right)\right]},
\]
steps, where $1-p$ lower bounds the probability that $\widetilde{\Theta}$ satisfies \cref{eqn:condition}.
\end{thm}
\begin{proof}
Assumptions 1 and 4 imply that $|X(s+\delta s) - X(s)| \leq c_0X(s)$ by \cref{prop:h_change}. Combining this with Assumptions 2 and 3, we have
\begin{equation*}
    X(s+\delta s) \geq (1-c_0)X(s) \geq (1-c_0)\frac{\widetilde{X}(s)}{1 + 9c_0/10} \geq \frac{1-c_0}{1+c_0}x_\min.
\end{equation*}
Hence, by \cref{prop:hoeffding}, constructing a $\widetilde{\Theta}$ obeying \cref{eqn:condition}
with success probability $1-p$ requires
\begin{equation} \label{eqn:n}
n = \ceil*{\left(\frac{1+c_0}{1-c_0}\right)^2\frac{50(V-1)^2(\kappa-1)^2}{81x_\min^2 c_0^2}\ln\left(\frac{2}{p}\right)}
\end{equation}
samples of $\mathcal{V}$. Taking $x = \widetilde{X}$ in \cref{eqn:condition} and \cref{prop:root} implies that the zero $\widetilde{X}(s+\delta)$ of $\widetilde{\Theta}$ satisfies \begin{equation*} \label{tildeX_bound}
\abs*{\widetilde{X}(s+\delta s) - X(s+\delta s)} \leq \frac{9c_0}{10}\left(\frac{1-c_0}{1+c_0}\right)x_\min \leq \frac{9c_0}{10}X(s+\delta s).
\end{equation*}
It then follows from \cref{prop:interval} with $\epsilon_0 = 9c_0/10$ that $\widetilde{X}(s+\delta s) \in [ax_0(s), bx_0(s)]$, where $a = (1-9c_0/10)(1-c_0)/(1+c_0) \geq (1-c_0)^2/(1+c_0)$ and $b = (1+9c_0/10)(1+c_0)/(1-c_0) \geq (1+c_0)^2/(1-c_0)$.
Because of Assumption 3, \textproc{FindRoot} proceeds to find the zero of $\widetilde{\Theta}$. Using the above interval in \cref{prop:bisection} and setting $\epsilon_1 = c_0/(10 + 9c_0) > c_0/19$, we see that applying
\begin{equation*}
    k = \left\lceil\log_2\left\{ \frac{19}{c_0}\left[\left(\frac{1+c_0}{1-c_0}\right)^2 - 1\right]\right\} \right\rceil
\end{equation*}
iterations of the bisection method returns an $x_0(s+\delta s)$ such that \begin{equation*}
    \abs*{x_0(s+\delta s)-\widetilde{X}(s+\delta s)} \leq \frac{c_0}{10 + 9c_0}\widetilde{X}(s+\delta s).
\end{equation*}
Each iteration evaluates $\widetilde\Theta(s+\delta s, x)$ once, and each evaluation can take as many operations as are required to construct $\widetilde{\Theta}$, which is determined by the number of samples, $n$. Therefore, noting that $\epsilon_0 + \epsilon_1 + \epsilon_0\epsilon_1 = c_0$, an estimate $x_0(s+\delta s)$ satisfying
\[
    \abs*{x_0(s+\delta s) - X(s+\delta s)} \leq \abs*{x_0(s+\delta s) - \widetilde{X}(s+\delta s)} + \abs*{\widetilde{X}(s+\delta s) -X(s+\delta s)} \leq \epsilon_1 \widetilde{X}(s+\delta s) + \epsilon_0 X(s+\delta s) \leq (\epsilon_0 + \epsilon_1 + \epsilon_0\epsilon_1) X(s+\delta s) = c_0 X(s+\delta s)
\]
can be obtained using at most $kn$ steps.
\end{proof} 
Since $H(0) =L$ is independent of $W$, we know that $X(0) = V$ and Line \ref{alg:cg_oracle:x0} of \cref{alg:cg_oracle} sets $x_0(0) = X(0)$. Thus, the conditions of \cref{thm:oracle_construction} are satisfied for the base case $s=0$. By induction, if the result of the theorem holds at some $s$, then Assumptions 1, 2, and 5 are automatically satisfied at the next step $s+\delta s$ by choosing the step size $\delta s$ as in Assumption 4 and taking $n$ at least as large as that in \cref{eqn:n}. Consequently, the result of \cref{thm:oracle_construction} holds for all steps until $s$ is such that Assumption 3 is not true, at which point the \textproc{\textbf{if}} statement of \textproc{FindRoot} is executed, \textproc{FindRoot} returns 0, and \textproc{GetGap} proceeds to call the \textproc{FinishSchedule} function discussed in the next subsection. 

The \textproc{\textbf{if}} statement of \textproc{FindRoot} also finds an upper bound $S_\min$  on $s_\min$ (defined as in \cref{thm:linear-envelope}), which will be useful for estimating the gap in \textproc{FinishSchedule}. To see why the value assigned to $S_\min$ in Line \ref{alg:cg_oracle:sm} of \cref{alg:cg_oracle} indeed upper bounds $s_\min$, suppose that $s$ is the first point for which Assumption 3 does not hold, so that \textproc{FindRoot} is not used to estimate $X(s+\delta s)$ and instead executes Lines \ref{alg:cg_oracle:sm} and \ref{alg:cg_oracle:sm_end}. If $s \leq s_\min$, then $h_{\{m\}} = g_{\{m\}}$ (by definition of $s_\min$) and by \cref{cor:cheeger},
\begin{equation*} \gamma(s) = (1-s)\gamma(G(s)) \leq (1-s)2g_{\{m\}}(s) = 2(1-s)(X(s) - 1) \leq 2(1-s)X(s).
\end{equation*}
Substituting this upper bound on $\gamma(s)$ into \cref{cor:smin}, it follows that
\begin{equation*}
    s_\min \leq s + \frac{4(1-s)X(s)}{\chi V}.
\end{equation*}
We know from Theorem 8 that the previous call to \textproc{FindRoot} returned an $x_0(s)$ such that $\abs{x_0(s) - X(s)} \leq c_0X(s)$.
Line \ref{alg:cg_oracle:sm} is executed either when $x(s) \leq x_\min$, in which case 
\begin{equation*}
    X(s) \leq \frac{x_0(s)}{1-c_0} \leq \frac{x_\min}{1-c_0},
\end{equation*} 
or because $\widetilde{X}(s+\delta s) < x_\min$. Suppose that the latter condition is satisfied but the former is not, i.e., that $\widetilde{X}(s+\delta s) < x_\min$ but $x(s) > x_\min$. In this case, note that we would have \begin{equation*}
    X(s+\delta s) \geq (1-c_0)X(s) \geq \frac{1-c_0}{1+c_0}x(s) \geq \frac{1-c_0}{1+c_0}x_\min,
\end{equation*}
which means that having chosen $n$ according to \cref{eqn:n}, $\abs{\widetilde{X}(s+\delta s) - X(s+\delta s)} \leq c_0X(s+\delta s)$, whence
\begin{equation*}
    X(s) \leq \frac{X(s+\delta s)}{1-c_0} \leq \frac{\widetilde{X}(s+\delta s)}{(1-c_0)^2} \leq \frac{x_\min}{(1-c_0)^2}.
\end{equation*}
Therefore, in either case, \cref{cor:smin} guarantees that 
\begin{equation*} \label{Smin}
    s_\min \leq s+\dfrac{4(1-s)x_\min}{(1-c_0)^2\chi V},
\end{equation*}
and \textproc{FindRoot} sets $S_\min$ accordingly. 
If, on the other hand, $s > s_\min$, the above inequality is trivially true, and so $S_\min$ is a valid upper bound in either scenario.

Finally, it is important to note that since $\widetilde{\Theta}$ is constructed using the same set of vertex samples at every step throughout a given run of BAA, the probability of failure $p$ from \cref{prop:hoeffding} is \emph{not} compounded at each step. After setting $n$ and randomly choosing $n$ vertex samples at the very first call BAA makes to \textproc{FindRoot}, the resultant $\widetilde\Theta$ either approximates ${\Theta}$ for all subsequent $s$ at which its root is used to approximate $X$, or never does. Thus, the entire procedure evolving from $s=0$ to $s=1$ succeeds with probability at least $1-p$.

\subsubsection{The \textproc{FinishSchedule} function}

The \textproc{FinishSchedule} function exploits the the lower bound of \cref{thm:linear-envelope} when \textproc{GetGap} is called for an $s$ for which $\widetilde{\Theta}$ cannot be used to reliably return an estimate of the gap $\gamma(s)$. Its behavior is simple: first, it underestimates the gap and returns its absolute lower bound, and then, once we know that we have definitely passed the minimum gap, it follows the linear envelope of \cref{thm:linear-envelope} until $s=1$. One could use computational basis measurements to avoid following the linear envelope and instead estimate $h_{\{m\}}$ using \cref{prop:large_side}; however whenever $\kappa = \bigO{1}$, doing so results in no asymptotic advantage. Thus, for simplicity, we follow the envelope itself.

Near the minimum gap, this procedure has the possibility of introducing additional steps into \cref{alg:BAA}, since in this intermediate region \textproc{GetGap} does not satisfy the conditions of \cref{thm:queries}. The following proposition determines the number of extra queries to \textproc{GetGap} introduced by this modified behavior.

\begin{prop}\label{prop:extra_steps_redux}\label{prop:extra_steps}
    With all quantities as in \cref{alg:cg_oracle}, if $x_\min < \chi V$, \cref{alg:BAA} requires at most 
    \[ 
        \bigO{\frac{x_\min}{\sqrt{V}}+\log_2\left(\sqrt{V}\right)}
    \] calls to \textproc{FinishSchedule}.
\end{prop}
\begin{proof}
Let $I_k' = \{s > s_\min \: | \: s \in I_k\}$ and $\mathcal{S}_\min = [s_\min,S_\min]$. Now, by \cref{alg:cg_oracle}, $\mu(\mathcal{S}_\min) = (1+c_0)(1-s_\min)\frac{2x_\min}{\chi (V-1)} = \bigO{\frac{x_\min}{V}}$. Thus, the number of queries of \cref{alg:cg_oracle} to Line \ref{alg:cg_oracle:smin} is $\bigO{\frac{x_\min \lambda_\max}{V \Gamma_\min}} = \bigO{\frac{x_\min}{\sqrt{V}}}$. Now, we will apply \cref{thm:queries_redux} with $R=1$ to the remaining points in $I_k'\setminus \mathcal{S}_\min$. 

    Consider $\abs*{s-s_\min}$ such that $\Gamma(s) \leq V 2^{-k}$. Then, by \cref{thm:linear-envelope}
    \begin{dgroup*}
    \[
        V 2^{-k} \geq \frac{(V-2)\chi}{4\kappa^5}\abs*{s-s_\min}
    \]
    \[
        \frac{4\kappa^5}{(V-2)\chi}V 2^{-k} \geq \abs*{s-s_\min}
        = \mu(I_k'\setminus \mathcal{S}_\min)
    \]
    \end{dgroup*}
    Thus, $\frac{8\kappa^5}{(V-2)\chi}V 2^{-k} \geq \mu(I_k\setminus \mathcal{S}_\min) = \bigO{2^{-k}}$ and we satisfy \cref{thm:queries_redux}. Hence, \textproc{FinishSchedule} requires at most $\bigO{\frac{x_\min}{\sqrt{V}}+\log_2\left(\sqrt{V}\right)}$ queries to \textproc{GetGap}.
\end{proof}

\subsubsection{Runtime}\label{sec:runtime}
Here, we bound the runtime of \cref{alg:BAA} were it to use the gap oracle constructed in \cref{alg:cg_oracle}. 

\begin{thm}\label{thm:full_runtime}
For $s \in [0,1]$, let $P(s)$ denote the projector onto the ground state of $H(s) = (1-s)L + sW$, and suppose that $\chi$ and $\kappa \geq W_{u_{V-1}}/W_{u_1} $ are constants independent of $V$. If \cref{alg:BAA} uses \cref{alg:cg_oracle} as the gap oracle, then for any $\epsilon >0$, it returns a state $\widetilde{P}(1)$ such that
\[ \norm*{\widetilde{P}(s) - P(s)} \leq \bigO{\epsilon}\]
with probability at least $1 - e^{-1/\epsilon}$ 
in time 
\[ \bigO{\frac{1}{\epsilon}\left(\sqrt{V} + (\kappa-1)^{2/3}V^{2/3}\log(\sqrt{V})\right)}.\]
\end{thm}
\begin{proof}
This proof proceeds in parts. First, we determine the total runtime due to calls to \textproc{GetGap}. Then, we find the total runtime of the procedure \textproc{GenerateState}. These runtimes are additive, so we combine them to determine the full runtime of \cref{alg:BAA} with oracle \cref{alg:BAA_oracle}.

\paragraph{Runtime due to \textproc{GetGap}}
We see from Lines \ref{alg:cg_oracle:n} and \ref{alg:cg_oracle:w0} of \cref{alg:cg_oracle} that constructing a function $\widetilde\Theta$ that well approximates $\Theta$ (in the sense of \cref{prop:hoeffding}) requires $n=\bigO{1 + {V^2(\kappa-1)^2}x_\min^{-2}\ln(1/p)}$ samples, conditioned on the vertex $m$ never being sampled. Whenever $m$ is sampled, Line \ref{alg:cg_oracle:w0} resamples until a different vertex is chosen. The probability of sampling $m$ is $1/V$ (assuming a uniform distribution over the vertices), so the probability that, for an integer $k\geq 0$, Line \ref{alg:cg_oracle:w0} loops more than $k$ times when it is reached is $V^{-k}$. Since Line \ref{alg:cg_oracle:w0} is reached $n$ times, the probability that every repetition of Line \ref{alg:cg_oracle:w0} requires more than $k$ operations is $V^{-kn}$. Therefore, the probability that populating $\mathcal{W}$ in Line  \ref{alg:cg_oracle:populate} requires no more than $kn$ steps and that the resultant $\widetilde\Theta$ is a good approximation for $\widetilde\Theta$ is at least $(1-V^{-kn})(1-p) \geq 1-p-V^{-kn}$. Since $V^{-kn}$ is doubly exponentially small, we can choose the upper bound on the failure probability to be a constant. Note from Line  \ref{alg:cg_oracle:populate} that the global array $\mathcal{W}$ is populated only once in a given run of \cref{alg:BAA}. Once $\mathcal{W}$ is fixed, each call to $\widetilde\Theta$ simply executes Line \ref{alg:cg_oracle:17} and returns.

\textproc{GetGap} uses \textproc{FindRoot} at each step, starting from $s=0$, until the \textbf{if} statement of \textproc{FindRoot} is executed. It is clear from Line \ref{alg:cg_oracle:if} that at every point $s$ for which the \textbf{if} statement is not satisfied, it must be the case that $x_0(s) > x_\min$, which implies that $X(s) \geq x_0(s)/(1+c_0) \geq x_\min/(1+c_0)$. Since $x_\min \geq 2(1+c_0)\sqrt{V}$ by Line \ref{alg:cg_oracle:xmin} whereas $X(s) \leq \sqrt{V-1} + 1 \leq 2\sqrt{V}$ at $s= s_\min$ by \cref{prop:norm_bound}, it follows that $s < s_\min$ for every point $s$ at which the root-finding procedure of \textproc{FindRoot} is executed, with the possible exception of the very last point. Thus, the gap estimate returned by Line \ref{alg:cg_oracle:start_end} of \textproc{GetGap} at these points is $\Theta(\Gamma(s))$, where $\Gamma(s)$ is defined in \cref{thm:linear-envelope}. \cref{thm:queries,thm:linear-envelope} then imply that the number of calls  \cref{alg:BAA} makes to \textproc{FindRoot} is $\bigO{\log\left(\sqrt{V}\right)}$.

Combining this with the with the number of steps required for each call to \textproc{FindRoot} given by \cref{thm:oracle_construction}, it follows that $\mathcal{O}(1+V^2\log(\sqrt{V})(\kappa -1)^2x_\min^{-2}\ln(1/p))$ operations are performed before \textproc{GetGap} starts using \textproc{FinishSchedule} instead of \textproc{FindRoot} to estimate the gap.
\cref{prop:extra_steps} shows that \textproc{FinishSchedule} is called $\mathcal{O}(x_\min/\sqrt{V} + \log(\sqrt{V}))$ times, and we see from \cref{alg:cg_oracle} that each call performs a single elementary computation. Adding this to the steps required by \textproc{FindRoot}, the runtime due to \cref{alg:BAA} calling \textproc{GetGap} is
\begin{equation*}
    \bigO{\frac{V^2\log(\sqrt{V})(\kappa-1)^2}{x_\min^2}\log\left(\frac{1}{p}\right) + \frac{x_\min}{\sqrt{V}} + \log\left(\sqrt{V}\right)}.
\end{equation*}

\paragraph{Runtime due to \textproc{GenerateState}} After completing the main loop of \cref{alg:BAA}, the adiabatic state at $s=1$ is prepared using \textproc{GenerateState}($\vec\gamma$,$0$,$\epsilon$). Taking $T_i = \mathcal{O}(\epsilon^{-1}\lambda_\max/\gamma_i)$ for each iteration of the main loop of \cref{alg:adiabatic} produces a state $\norm{\widetilde{P}(s) - P(s)} \leq \bigO{\epsilon}$ by \cref{thm:adiabatic5,thm:adiabatic_final}. For each $s_i$ such that $s_i \in [s_\min,S_\min]$, the algorithm sets $\lambda_\max/\gamma_i = \bigO{\sqrt{V}}$. Thus, for each such $i$, \textproc{GenerateState} adds time $\bigO{\sqrt{V}/\epsilon}$ to the overall runtime. By \cref{prop:extra_steps_redux}, we know that while $s \in [s_\min, S_\min]$ we introduce at most an additional $\bigO{\frac{x_\min}{\sqrt{V}}}$ checkpoints. Thus, this portion of the domain introduces a total overhead of $\bigO{\frac{x_\min}{\sqrt{V}}\frac{\sqrt{V}}{\epsilon}} = \bigO{\frac{x_\min}{\epsilon}}$. For all $s \notin [s_\min, S_\min]$, \cref{thm:adiabatic_final} applies directly. Thus, we arrive at a total runtime of $\bigO{\frac{1}{\epsilon}\left(x_\min + \sqrt{V} \right)} = \bigO{\frac{x_\min}{\epsilon}}$, where we note that from Line \ref{alg:cg_oracle:xmin} that $x_\min \geq \sqrt{V}$.

Since calls to \textproc{GetGap} and \textproc{GenerateState} are additive and we assume $\epsilon < 1$, the total time of the algorithm is 
\[
    \bigO{\frac{x_\min}{\epsilon} + (\kappa-1)^2\frac{{V^{2}}\log(\sqrt{V})}{x_\min^2}\ln(1/p)}.
\]
Letting $p = \Omega(e^{-1/\epsilon})$ yields a total runtime of
\[
\bigO{\frac{1}{\epsilon}\left[x_\min + (\kappa-1)^2\frac{V^{2}}{x_\min^2}\log(\sqrt{V})\right]}.
\]
Noting that when $x_\min = \sqrt{V}$, the definition of $x_\min$ in \cref{alg:cg_oracle} gives
\begin{dgroup*}
\[
    x_\min \geq (\kappa-1)^{2/3}(V-1)^{2/3}
\]
\[
    x_\min^3 \geq (\kappa-1)^2 (V-1)^2
\]
\[
    x_\min \geq (\kappa - 1)^2 \frac{(V-1)^2}{x_\min^2}.
\]
\end{dgroup*}
Thus, since, $\Omega(\sqrt{V})= x_\min = \bigO{1+(\kappa-1)^{2/3}V^{2/3}}$, we have a total runtime of
\[
\bigO{\frac{1}{\epsilon}\left[x_\min + (\kappa-1)^2\frac{V^{2}}{x_\min^2} \log(\sqrt{V}) \right]} = \bigO{\frac{1}{\epsilon}\left(\sqrt{V} + (\kappa-1)^{2/3}V^{2/3}\log(\sqrt{V})\right)}.
\]
\end{proof} 

\subsection{Optimization}
We now present an optimization algorithm, \cref{alg:optimal}, that optimizes the runtime of \cref{alg:BAA} using multiple copies of the gap oracle \textproc{GetGap} of \cref{alg:cg_oracle} when the spectral ratio $\kappa$ of the cost function $W$ is unknown. It proceeds by guessing that $\kappa \approx 1$ and increases $\kappa$ until we are guaranteed that $\kappa$ upper bounds the spectral ratio of $W$. 

\begin{algorithm}[H]
 \caption{\label{alg:optimal}Optimize}
 \begin{algorithmic}[1]
    \Require $\lambda_\max \geq \norm{H}$, the spectral gap $\gamma(0)$ of $H(0)$, a lower bound $\gamma(1)$ on the spectral gap of $H(1)$, an oracle $\textproc{GetGap}_\kappa$ that depends upon the parameter $\kappa$, a constant $p \in (0,1)$ independent of $V$
 \Function{Optimize}{$W$}
 \State $\delta \gets \frac{3}{2}\log_V(\frac{3}{2})$   \Comment{Step size by \cref{thm:search}}
 \State Choose $N \gets \frac{\log(pV^{-1/6})}{\log((1+e^{-1})\epsilon)}$
 \For {$i \in \llbracket \frac{1}{4\delta} \rrbracket$}\label{alg:optimal:loop}
    \State $\kappa \gets 1+V^{i\delta - \frac{1}{4}}$. \Comment{Guess $\kappa(W) \leq \kappa$}
    \State $\Psi \gets \left[\Call{BAA}{\textproc{GetGap}_\kappa}\right]_{i=0}^{N}$ \Comment{Collect the results of BAA}
    \State $\Psi \gets$ \Call{Measure}{$\Psi$} \Comment{Measure in the computational basis}
    \If{$W_b = 0$ for any $b \in \Psi$} \Return $\ket{b}\bra{b}$ \EndIf
 \EndFor
 \State $\kappa \gets \lambda_\max/\gamma(1)$ \Comment{Assume worst case $\kappa$}
 \State \Return $\Call{BAA}{\textproc{GetGap}_\kappa}$ \Comment{Run BAA with worst case $\kappa$}
 \EndFunction
 \end{algorithmic}
\end{algorithm}
 
The following theorem shows that we can perform \cref{alg:optimal} and achieve the runtime bound of \cref{thm:full_runtime} for $\kappa \approx \kappa(W)$, with only logarithmic asymptotic overhead.

\begin{thm}\label{thm:search}
     For $\kappa = \frac{W_{u_{V-1}}}{W_{u_1}} \leq \frac{\lambda_\max}{\gamma(1)}$, \cref{alg:optimal} has an expected runtime of
    \[
        \bigO{\frac{\log^2(\sqrt{V})}{\epsilon\log(1/\epsilon)}\left(\sqrt{V} + (\kappa-1)^{2/3}V^{2/3}\right)}
    \]
    and returns $\ket{m}\bra{m}$ with probability greater than $1 - pV^{-1/6}$.
\end{thm}
\begin{proof}
    Assume that $\kappa = 1+V^{-x}$ and $\kappa_j = 1 + V^{j \delta - \frac{1}{4}}$. Then, there exists a first $j$ such that $x_j = \frac{1}{4} - j\delta \leq x \leq \frac{1}{4} - (j-1)\delta$
        \[
            \abs*{\kappa_j -\kappa} = \abs*{V^{j \delta - \frac{1}{4}} - V^{-x}}
            = V^{-x}\abs*{V^{x + j\delta - \frac{1}{4}}-1}
            = V^{\delta-x}\abs*{V^{x + (j-1)\delta - \frac{1}{4}}-V^{-\delta}}
            \leq (\kappa-1) V^{\delta}.
        \]
    Taking $\delta = \frac{3}{2}\log_V(\frac{3}{2})$, we have that for this $j$,
    \begin{align*}
        \kappa_j -\kappa &= (\kappa-1) \left(\frac{3}{2}\right)^{3/2}.
    \end{align*}
    Thus, for some choice of $j$, we have that $\kappa_j \geq \kappa$ and $\kappa_j-1 = \Theta(\kappa-1)$. Thus, assuming that \cref{alg:optimal} has not returned prior to $j$, we wish to evaluate
    \[
        \sum_{i=0}^{j} (\kappa_i-1)^{\frac23}V^{\frac{2}{3}} =\sum_{i=0}^{j} V^{\frac{2}{3}\left(1-\frac{1}{4}+ i\delta\right)}
        =\sqrt{V}\sum_{i=0}^{j} V^{\frac{2i}{3}\delta}
        =\sqrt{V} \frac{V^{\frac{2\delta}{3}(j+1)}-1}{V^{\frac{2\delta}{3}}-1}
        \leq 2\sqrt{V}V^{\frac{2j \delta}{3}}
        = 2 V^{\frac{2}{3}} V^{\frac{2}{3}\left(j\delta -\frac{1}{4}\right)}
        = 2 V^{\frac{2}{3}}\left(\kappa_j - 1\right)^{\frac23}
        \leq 2(\kappa-1)^{\frac23} V^{\frac{2}{3}}\left(1+ V^{2\delta}\right)
        = \frac{27}{4}(\kappa-1)^{\frac23} V^{\frac{2}{3}}.
    \]
    Thus, by \cref{thm:full_runtime} we can reach the end of the $j$\textsuperscript{th} iteration of \cref{alg:optimal} Step \ref{alg:optimal:loop} in total time 
    \[
        \bigO{\frac{N}{\epsilon}\left(\sqrt{V} + (\kappa-1)^{2/3}V^{2/3}\log(\sqrt{V})\right)}.
    \] On the $j$\textsuperscript{th} iteration, we are guaranteed to satisfy the conditions of \cref{thm:full_runtime} and thus we produce $N$ states such that $\widetilde{P} = \ket{m}\bra{m} + \bigO{\epsilon}$, each with probability at least $1-e^{-1/\epsilon}$. After measuring each in the computational basis $N$ times, we do not return $\ket{m}\bra{m}$ with probability at most 
    \[
        p^N = \left(1-\left((1-\epsilon)(1-e^{-1/\epsilon})\right)\right)^N
        \leq \left((1+e^{-1})\epsilon\right)^N.
    \]

    Now, let $f$ be the the greatest value of $i$ reached by \cref{alg:optimal} Step \ref{alg:optimal:loop}. Let $NT_i$ be the total time taken through the end of the $i$\textsuperscript{th} iteration of \cref{alg:optimal} Step \ref{alg:optimal:loop}. Then, 
\[
    \E{T_f} = \sum_i T_i {\Pr{f=i}}
    \leq T_j + N\sum_{i > j}{\Pr{f=i}}\bigO{\frac{V^{\frac{2}{3}\left(1-\frac{1}{4}+ i\delta\right)}}{\epsilon}\log(\sqrt{V})}
    \leq T_j + N\sum_{i >j}\bigO{\frac{V^{\frac{2}{3}\left(1-\frac{1}{4}+ i\delta\right)}}{\epsilon}\log(\sqrt{V})} \left((1+e^{-1})\epsilon\right)^{\log_{(1+e^{-1})\epsilon}\left(pV^{-1/6}\right)}
    = T_j + N\sum_{i > j} \bigO{\frac{V^{\frac{2}{3}\left(1-\frac{1}{4}+ i\delta\right)}}{\epsilon}\log(\sqrt{V})}\left(\frac{p}{V^{1/6}}\right)
    = T_j +  \bigO{N\sum_{i > j}\frac{V^{\frac{1}{3}(1+2i\delta)}}{\epsilon}\log(\sqrt{V})}
    = T_j + \bigO{N\frac{\sqrt{V}}{\epsilon}\log(\sqrt{V})}
    = \bigO{\frac{N}{\epsilon}\left(\sqrt{V} + (\kappa-1)^{2/3}V^{2/3} \log(\sqrt{V}) \right)} + \bigO{N\frac{\sqrt{V}}{\epsilon}\log(\sqrt{V})}
    = \bigO{\frac{\log^2(\sqrt{V})}{\epsilon\log(1/\epsilon)}\left(\sqrt{V} + (\kappa-1)^{2/3}V^{2/3}\right)}.
\]

\end{proof} 
\section{Discussion and future work}\label{sec:discussion}

\subsection{Improving the Cheeger inequality}
    The bound derived in \cref{thm:cheeger} is almost certainly loose by a factor of approximately $\kappa^2$, however a tighter version of this theorem remains elusive. Since improving the inequality provides no asymptotic advantage, we did not attempt to achieve a nearly optimal inequality in the present work. Nonetheless, tight inequalities are always of mathematical interest and deriving them is well-motivated, especially if one were to look to apply this algorithm in practice.
    
    The most likely approach to successfully achieving an inequality that scales like
    \[
        {2h_{\{m\}} \geq \gamma \gtrsim \frac{h_{\{m\}}}{\kappa}    }
    \]
    would be to consider the ratios of the components of $\phi_1$. That is, if $W_m < W_{u_1} \leq \dots W_{u_{V-1}}$, we presently cannot get a useful bound on 
    \[
        {\frac{\phi_1(u_{V-1})}{\phi_1(u_1)} = \frac{V+W_{u_1}-\lambda_1}{V+W_{u_{V-1}}-\lambda_1} }
    \]
    as the numerator can tend towards $0$ for large enough $W$. Nonetheless, if we consider $\kappa_i = \frac{W_{u_{V-1}}}{W_{u_{i}}}$
    \[
        {\frac{\phi_1(u_{V-1})}{\phi_1(u_{j>1})} = \frac{V+W_{u_j}-\lambda_1}{V+W_{u_{V-1}}-\lambda_1} \geq \frac{\kappa_j-1}{\kappa_1-1}}. 
    \]
    Thus, it would likely be possible to derive an inequality that keeps \cref{thm:cheeger} tighter, incorporating more information about $\phi_1$ in a reasonable way. This seems to require an appropriate modulus and we leave this as a separate technical project.

\subsection{Shaving off extra factors of \texorpdfstring{$\kappa$}{κ}}
    Note that we reach a state such that $\phi(m)^2 \geq 1/2$ before our Cheeger inequality \cref{thm:cheeger} becomes weak by factors of $\kappa$. Thus, if we are only looking to perform optimization and factors of $\kappa^5$ look ominous, we could stop our algorithm short and incur no factors of $\kappa$ other than insofar as they improve \cref{prop:hoeffding}. That is, we could prepare a state such that $\norm*{\ket{m} - \ket\phi} \leq \frac{1}{2}+ \bigO{\epsilon}$. By repeating this procedure $\log_2(1/\epsilon)$ times, we would then be able to return $\ket{m}$ with probability $1-\bigO{\epsilon}$. Everything else would remain unchanged.

\subsection{Multiple marked states}
    We should consider the case that there is more than one state $M=\{m_i\}_{i=0}^{k-1}$ such that $W_{m_i} = 0$. For this, \cref{thm:cheeger} would need to be modified to handle the degeneracy, which could be done by first projecting into the subspace that identifies all marked states as a single state. 
    
    Since in the restricted subspace, if we let $\psi_i(M) = \sum_{m\in M}\phi_i(m)$ and $\psi_i(v) = \phi_i(v)$, for any eigenvector $\phi$ corresponding to eigenvalue $\lambda$, we have
    \[
        (V+W_u -\lambda)\phi(u) = \sum_{v\neq u}\left(\phi(v)-\phi(u)\right)
    \]
    we have that
    \[
        (V-\lambda)\psi(M) = \sum_{m\in M}\sum_{v}\left(\phi(v)-\phi(m)\right) 
        = \sum_{m\in M}\sum_{v\notin M}\left(\phi(v) - \phi(m)\right) + \sum_{m\in M}\sum_{\substack{v \in M}}\left(\phi(v) - \phi(m)\right)
        = \sum_{v\notin M}\left(k\phi(v) - \psi(M)\right).
    \]
    Then,
    \[
        (\lambda_{k}-\lambda_0)\psi_0(M)\psi_{k}(M) = k\sum_{v\notin M}\left(\phi_0(v)\psi_{k}(M) - \phi_{k}(v)\psi_0(M)\right).
    \]
    Comparing this to \cref{thm:cheeger}, one could clearly derive an appropriate Cheeger inequality that applies to the relevant subspace, where we are interested in the spectral gap $\lambda_{k} - \lambda_0$. One would need to take care, however, since $\psi_0(M)^2 \neq \sum_{m\in M}\phi_0(m)^2$. Thus, the improvement would not be the simple factor of $k$ that we see above. The tighter bound would ultimately result from the reduced number of vertices under consideration in the equivalent of \cref{prop:large_side}, using the fact that we would now be considering the point $s_\min$ as the point at which $\sum_{m \in M} \phi(m)^2 = \frac{1}{2}$.

    One could then simply proceed by calling \cref{alg:BAA} assuming that there are $V/2^i$ marked states for $i=0,1,2,\dots,V-1$ until some state $u$ such that $W_u=0$ is returned. A similar approach was used by two of the authors in \cite{Wan} in a completely different context, but the same technique should apply here and achieve optimal scaling. However, combining this with \cref{alg:optimal} would require carefully balancing parameters. It is less clear whether the fixed point methods of \cite{dalzell2017fixed} would work with BAA, though their adaptation would certainly be interesting and the claim that similar methods work in \cite{Boixo2009} suggests they might be promising. 
    
\subsection{Designing optimization hardware}\label{subsec:cuts}
    Our results suggest that, for near-term quantum hardware, one should look to create driving Hamiltonians $L$ such that the Cheeger constant, or something similar to it, is easy to evaluate for arbitrary cost functions. It is worth noting that, because the Cheeger constant both upper and lower bounds the spectral gap, any method that is capable of optimizing an annealing schedule would be equivalent to our procedure.    
    
    \cref{alg:cg_oracle} demonstrates that at least under some circumstances and for particular graphs, this counterintutively appears possible. Importantly, although the spectral gap may in general be (very) hard to estimate, the points at which it becomes difficult can possibly be ignored, much like in \cref{alg:cg_oracle}. Furthermore, the Cheeger constant itself has a more physical meaning: the numerator in the Cheeger constant itself is a measurement of the energy of some \textit{cut operator}
    \[
        C_S = \frac{1}{2}\sum_{\substack{i \in S \\ j \notin S}}\left(\ket{i}\bra{j}+\ket{j}\bra{i} \right).
    \]
    That is,
    \[
        g_S =\frac{\bra{\phi}C_S \ket{\phi}}{\sum_{u \in S} \phi(u)^2}
    \]
    and
    \[
        h_S = \max_{S' \in \{S,\overline{S}\}}\frac{\bra{\phi}C_S \ket{\phi}}{\sum_{u \in S} \phi(u)^2}.
    \]
    If we prepare $\ket\phi$ to high enough accuracy, then we should expect that both the numerator and denominator can be measured. That is, the denominator is just the probability that $\phi$ measured in the computational basis will be found in $S$ and the numerator is the energy of the particular cut which, at least in our complete graph case, can be measured in the $X$-basis. Importantly, these bases remain fixed throughout the interpolation and so we need no knowledge of instantaneous eigenbases to determine each $g_S$. Furthermore, for many graphs, we can probably find an appropriate set of cuts $\{C_S\}$ such that each $S$ is of different size and we are able to permute $W$ to measure sets corresponding to different elements of the cost function. Thus, with only a fixed number of physical cuts, we can create a much larger number of computational cuts. The idea is sketched in \cref{alg:BAA_oracle}.
    
\begin{algorithm}
 \caption{\label{alg:BAA_oracle}General oracle}
 \begin{algorithmic}[1]
    \Require A set of cuts $\mathcal{S}$, the cost function $W$, a number $N$, a known lower bound on the gap $\gamma_\min$
 \item[]
 \Function{GetGap}{$s,\delta s,\gamma,c_0$}
    \If{$h$ is efficient to compute}
        $h \gets (1+c_0)\gamma$ \Comment{Upper bound next $h$}
        \For{$S \subset \mathcal{S}$}
                \State $P_S \gets \sum_{i \in S} \ket{i}\bra{i}$ 
                \State $P_i \gets \left[\Call{GenerateState}{s,\delta s, \gamma,     c_0}\right]_{i=0}^{2N-1}$ \Comment{Generate some projectors at $s+\delta s$}
                \State $E_{C_S} \gets \frac{1}{N}\sum_{i=0}^{N-1}\mathrm{tr}\left(C_S P_i\right)$ \Comment{Measure the energy of the cut}
                \State $p_S \gets \frac{1}{N}\sum_{i=N-1}^{2N-1}\mathrm{tr}\left(P_S P_i\right)$ \Comment{Measure the probability of being in the cut}
                \State $h \gets \min\left(h,\frac{E_{C_S}}{\min(p_S,1-p_S)}\right)$
        \EndFor
    \State        \Return $h/2$ \Comment{Return a lower bound on the gap to within a constant factor}
    \EndIf
 \State \Return $\gamma_\min$ \Comment{Return a known lower bound on the gap}
 \EndFunction    
 \end{algorithmic}
\end{algorithm}
     
    In fact, this oracle is the reason BAA is named such; after sampling a state $\delta s$ away from what we presently know how to prepare, the oracle requires that we start our adiabatic procedure over entirely. Hence, we creep along slowly, only ever advancing by $\delta s$ in a given step. 
    
    For general graphs, using adiabatic processes as subroutines of whatever algorithm takes the place of our oracle may not just be useful, but is probably necessary. (In fact, that we did not need to use it in the present work came as a surprise to the authors.) If we satisfy something like \cref{thm:queries_redux}, then we are guaranteed that we need to repeat the adiabatic procedure at most $\bigO{\log(1/\gamma_\min)}$ times and, thus, the need for restarts should not be concerning.
    
\subsection{Improving the oracle}
    Although sufficient, \cref{alg:cg_oracle} can be greatly improved, especially after reaching the point labeled $s_\min$. One method for doing so would be to sample $\phi(m)^2$ by letting \cref{alg:cg_oracle} call \cref{alg:adiabatic} as discussed above, performing computational basis measurements on the result, and then exploiting \cref{prop:norm_bound}. \cref{prop:norm_bound} also guarantees that $\phi(m)^2> \epsilon$ for some constant $\epsilon<1/2$ even when $s < s_\min$. Thus, in some places, one might also be able to improve the provided bounds by multiple factors of $\kappa$, by switching to a computational basis measurement procedure \emph{prior} to $s=s_\min$. In the present context, none of these changes would achieve better asymptotic scaling and therefore they are not pursued. Nonetheless, for practical applications, optimizing constants and appropriately balancing some of the procedures presented might be important. 
    
\subsection{Removing the restriction \texorpdfstring{$W_m = 0$}{Wm = 0}}
    Although the existence of this condition allows us to solve the decision problem of whether $0 \in W$,\footnote{Just assume that $W_m =  0$, apply \cref{alg:BAA}, and if you get $0$ at the end of the day respond ``yes''.} it is not sufficient to arbitrarily optimize a set. For this, we may need to introduce cut operators like those discussed in \cref{subsec:cuts}, improve the classical parts of our oracle, or introduce heuristics to guess at the minimum $W_m$. Presently, if we think of $W:V \longrightarrow [0,1]$, our algorithm is actually flexible enough to find $m$, provided that we know that there exists a $\widetilde{W}$ such that $\abs*{\widetilde{W} - W_m} \leq \epsilon V^{-1/3}$. Of course, this is exponentially small on the scale of the problem and we would prefer to generalize beyond the current oracle for full optimization. Whether this can be done while remaining restricted to computational basis measurements is presently unclear.
    
\subsection{The role of other paths}\label{sec:other_paths}
    Interpolated Hamiltonians of the form
    \[
        H_i(s) = (1-s)H(0) + s(1-s)\widetilde{H}_i + sH(1)
    \]
    where $\widetilde{H}_0 = 0$ and $\widetilde{H}_i$ is some arbitrary time-independent Hamiltonian are often used in the hope that $\gamma_\min\left(H_i\right) \geq \gamma_\min\left(H_0\right)$, speeding up adiabatic processes. However, numerical results suggest that doing so usually also increases the width of the corresponding gap minimum. In the language of \cref{thm:queries_redux}, if $I_k^{(i)} = \{s : \gamma(H_i(s)) \leq \frac{\norm{H_i}}{2^k}\}$ and $\mu(I_k) \leq C 2^{-k}$, then we know that performing BAA on $H_0$ requires at most $\bigO{C\log_2\left(\frac{\norm{H}}{\gamma_\min(H_0)}\right)}$ queries to \textproc{GetGap}. Now, if replacing $H_0$ with $H_i(s)$ comes at cost that $\mu(I_k^{(i)}) \geq \frac{\gamma_\min(H_i)}{\gamma_\min(H_0)}\mu(I_k^{(0)})$, we know that whatever performance we might gain by being able to vary our Hamiltonian faster, we lose to an increased number of discretization points (in this case, the number of queries). 
    
    This is not to say that when the conjecture above holds true, intermediate Hamiltonians will not serve a purpose, they still might. However, the role of $H_i \neq H_0$ seems to be that these intermediate Hamiltonians may make queries to \textproc{GetGap} easier. That is, if we are to consider $H_i \neq H_0$, we should probably look to determine the class $H_i$ such that we can guarantee rapid returns from \textproc{GetGap}. It is possible that such an intermediate interpolation would even enhance the abilities of our oracle \cref{alg:cg_oracle}. Additionally, using \cref{alg:cg_oracle}, we might be able to use these intermediate Hamiltonians only when we know that queries are becoming more challenging, as is the case for $h\sim \sqrt{V}$ in \cref{alg:cg_oracle}.

\subsection{The width of the minimum gap}
    We believe that the mathematical project of studying the width of the minimum gap in an interpolated Hamiltonian (or the measures $\mu(I_k)$), especially in the context of the discussion of \cref{sec:other_paths}, would be quite interesting. Of course, this would be a question of pure analysis, but nonetheless well-motivated by BAA and \cref{thm:queries_redux}. Understanding this relationship would, in principle, be a key component of designing optimized schedules and also determining whether algorithmically easier and faster paths exist. Furthermore, such a bound would have complexity-theoretic implications. That is, if for interpolated Hamiltonians the width of the minimal gap can indeed always be bounded as in \cref{thm:queries_redux}, then we would know that for any driving Hamiltonian, either determining the gap is always hard for hard problems or there would be a tradeoff in the difficulty of determining the size gap and the size of the gap itself. Otherwise, \cref{alg:BAA} solves hard problems with bounded probability.
    
    To the knowledge of the authors, no appropriate, general inequalities yet exist.
    
\subsection{Classical algorithms}
    There is no reason that \textproc{GetGap} must correspond to \cref{alg:cg_oracle} or, for that matter, that the driving Hamiltonian should be restricted to one considered in this paper. It seems reasonable that one might use BAA as a tool for solving classical computation problems. Indeed, the authors of \cite{Boixo2009} suggest as much for their strategy. Since we know that simulated quantum annealing can potentially be faster than, say, simulated annealing, simulated quantum annealing using BAA might be faster still \cite{Crosson2016}. Furthermore, the strategy of BAA may be able to expedite eigensolvers by replacing heuristic guesses at eigenvalues with genuine approximations, expedite Monte Carlo methods for studying phase transitions of Hamiltonian systems by automatically focusing on relevant regions of parameter space, or be a useful component in other classical randomized approximation schemes. The overall approach need not be restricted to spectral information either, any information that can be rigorously bounded perturbatively can be utilized to provide step sizes in similar variational approachen.

\section{Acknowledgements}
The authors would like to thank Aram Harrow, Adrian Lupascu, Antonio Martinez, and Jon Yard for useful discussions. This work was completed in part while AL and KW were affiliated with, and BL was visiting, Perimeter Institute. \PIRA

\appendix

\section{Monotonicity of \texorpdfstring{$\phi(m)$}{φ(m)}}
\begin{prop}\label{prop:monotone} Let $\lambda(s)$ and $\phi>0$ denote the ground-state eigenvalue and normalised eigenvector of $G(s) = L + \dfrac{s}{1-s}W$ for $s \in [0,1)$, where $L$ is the combinatorial graph Laplacian of the complete graph on $V$ vertices and $W$ is a diagonal matrix. If $W_m =0$ is the unique smallest eigenvalue of $W$, then $d\phi(m)/ds > 0$, where $\phi(m)$ is the component of $\phi$ corresponding to the vertex $m$.
\end{prop}
\begin{proof}
By \cref{eqn:eigenvectors},
\begin{equation*}
    \frac{\phi(m)}{\phi(u)} = \frac{V+W_u(s) -\lambda(s)}{V-\lambda(s)} = 1 + \frac{W_u(s)}{V-\lambda(s)}
\end{equation*}
    for any $u\neq m$, where $W_u(s) = sW_u/(1-s)$. Taking the derivative,
\begin{align*}
    \frac{d}{ds}\left(\frac{\phi(m)}{\phi(u)}\right) &= \frac{d}{ds}\left(\frac{W_u(s)}{V-\lambda(s)}\right)\\
        &=\frac{1}{V-\lambda(s)}\left(\frac{dW_u(s)}{ds} + \frac{W_u(s)}{V-\lambda(s)} \frac{d \lambda(s)}{ds} \right) \\
        &=\frac{1}{V-\lambda(s)}\left[\frac{d}{ds}\left(\frac{s}{1-s}W_u\right) + \frac{W_u(s)}{V-\lambda(s)} \bra\phi \frac{d}{ds}\left(L + \frac{s}{1-s}W\right)\ket\phi \right] \\
        &= \frac{1}{V-\lambda(s)}\frac{1}{(1-s)^2}\left[W_u + \frac{W_u(s)}{V-\lambda(s)}\frac{1}{(1-s)^2} \bra\phi W \ket\phi \right].
    \end{align*}
$\lambda(s) \leq \bra{m}\left(L + \frac{s}{1-s}W\right)\ket{m} = \bra{m}L \ket{m} = V-1$, and $W_{u\neq m} >0$ by assumption, so $W_{u}(s), \bra{\phi}W\ket{\phi} > 0$. Thus,  $\dfrac{d}{ds}\left(\dfrac{\phi(m)}{\phi(u)}\right)>0$ for all $u\neq m$. This together with the fact that $\phi > 0$ implies that
\begin{equation*}
        \frac{d\phi(m)^2}{ds} > \frac{d\phi(u)^2}{ds}
\end{equation*}
for all $u \neq m$. 
Then, using the normalization condition $1=\sum_{u}\phi(u)^2$, we have
\begin{equation*}
    0 = \sum_{u}\frac{d\phi(u)^2}{ds} < \sum_{u} \frac{d\phi(m)^2}{ds} = V\frac{d\phi(m)^2}{ds}
\end{equation*}
or $d\phi(m)^2/ds >0$.
\end{proof}
 
 \printbibliography
\end{document}